\documentclass[journal,10pt]{IEEEtran}

\usepackage{cite}
\usepackage{color} 
\usepackage{stfloats}
\usepackage{booktabs} 
\usepackage[table]{xcolor}
\newcolumntype{C}[1]{>{\centering\arraybackslash}m{#1}}

\ifCLASSINFOpdf
   \usepackage[pdftex]{graphicx}

\else

\fi

\usepackage{cite}
\usepackage{amsmath,amsthm,amssymb,graphicx,algorithm,algorithmic}
\usepackage{makecell}
\usepackage{subfigure}
\usepackage{float}
\usepackage{hyperref}
\usepackage{booktabs}

\usepackage{enumerate}

\newtheorem{theorem}{\bf Theorem}

\newtheorem{lemma}{\bf Lemma}

\newtheorem{remark}{Remark}

\usepackage{color}   
\usepackage{array}
\usepackage{multicol}

\begin{document}

\title{Secure Coverage Enhancement in Aerial Reconfigurable Intelligent Surface-Assisted High-Speed Train Communication Systems}

\author{

Changzhu~Liu,~\IEEEmembership{Member,~IEEE},
Ruisi~He,~\IEEEmembership{Senior~Member,~IEEE},
Bo~Ai,~\IEEEmembership{Fellow,~IEEE}, \\
Yong Niu,~\IEEEmembership{Senior Member,~IEEE}, 
Zhu~Han,~\IEEEmembership{Fellow,~IEEE},
Gongpu~Wang,~\IEEEmembership{Member,~IEEE},\\
Haoxiang~Zhang,
Jiahui~Han,
and Zhangdui~Zhong,~\IEEEmembership{Fellow,~IEEE}

\thanks{

Changzhu Liu is with School of Intelligent Engineering and Intelligent Manufacturing, Hunan University of Technology and Business, Changsha 410205, China, and also with the School of Electronics and Information Engineering, Beijing Jiaotong University, Beijing 100044, China (e-mails: liuchangzhu@hutb.edu.cn)


Ruisi He, Bo Ai, Yong Niu and Zhangdui Zhong are with the School of Electronics and Information Engineering, Beijing Jiaotong University, Beijing 100044, China (e-mails: ruisi.he@bjtu.edu.cn; boai@bjtu.edu.cn; niuy11@163.com; zhdzhong@bjtu.edu.cn).

Gongpu Wang is with the School of Computer and Information Technology, Beijing Jiaotong University, Beijing 100044, China (e-mail: gpwang@bjtu.edu.cn).

Zhu Han is with the Department of Electrical and Computer Engineering at the University of Houston, Houston, TX 77004 USA (e-mail: hanzhu22@gmail.com).

Haoxiang Zhang and Jiahui Han are with the China Academy of Industrial Internet, Ministry of Industry and Information Technology, Beijing, China (zhx61778294@126.com; hjh1760708@126.com).

              }
              }

{}

\maketitle

\begin{abstract} %
High-speed trains (HSTs) have become a prominent means of transportation, requiring high data rates and reliable communication services for HST passengers. However, the wireless channels in HST communication systems are susceptible to various security threats, including eavesdropping. Addressing these security concerns is therefore of critical importance. One promising technology for enhancing security is the integration of a reconfigurable intelligent surface (RIS) on an unmanned aerial vehicle, referred to as an aerial reconfigurable intelligent surface (ARIS). This technology offers significant potential for improving wireless network performance, though it also introduces unique challenges in terms of physical layer security (PLS). This paper investigates the PLS of ARIS-aided HST communication systems. A problem of maximizing the weighted sum secrecy rate is formulated by jointly optimizing the active beamforming at the base station (BS) and the phase shift at the ARIS, subject to constrains on the BS transmit power and the unit modulus of the ARIS reflecting coefficient. To address this problem, a joint optimization algorithm is proposed using the block coordinate descent method. Specifically, the problem is decomposed into two subproblems: active beamforming design and ARIS phase shift optimization. The active beamforming is optimally designed via the successive convex approximation technique, while the ARIS phase shift is efficiently updated using the alternating direction method of multipliers technique. Simulation results demonstrate the rapid convergence of the proposed algorithm, which achieves a higher secrecy rate compared to existing methods in the literature.


\end{abstract}

%
\IEEEpeerreviewmaketitle

\section{Introduction}
\IEEEPARstart{O}{ver} the past decade, the swift development of high-speed trains (HSTs) have significantly improved travel convenience and transformed them into an essential transportation mode, owing to their high efficiency, reliability, exceptional safety, and energy-saving features \cite{c2new1,c2new2,c2new3,c2new4}. Concurrently, high-quality wireless communication with substantial data rates is crucial for providing passenger services on HSTs. The HST wireless communication signal system plays a crucial role in maintaining the safety of trains traveling at high speeds and densities. As HST wireless communication systems continue to evolve in terms of intelligence and information processing, the wireless networks within these systems have become more open, interconnected, and scalable. However, the inherent broadcasting characteristics of wireless signals mean that they can be intercepted within the coverage range, making these networks vulnerable to various security threats such as eavesdropping \cite{c3}. This makes the security of HST wireless networks an urgent challenge that needs to be addressed.

In recent years, reconfigurable intelligent surfaces (RIS) has garnered significant interest in the wireless communications sector as an innovative and revolutionary technology, also being recognized for its energy-efficient potential. An RIS consists of a plane made up of numerous low-cost, nearly passive reflective elements, each equipped with controllable phase shifting capacity \cite{c5}. These individually regulated and easily adjustable reflective elements enable deliberate modifications to the electromagnetic characteristics of the incident signal. Consequently, they can actively manipulate the wireless propagation, facilitating intelligent reflection that directs the signal to the intended receiver, thereby effectively enhancing the received signal strength \cite{c6}.

Recent developments in wireless communication networks have underscored the increasing demand for robust security solutions. Physical layer security (PLS) technique has emerged as a potentially effective approach to addressing security challenges in wireless communications \cite{c9,c10}. PLS leverages the intrinsic properties of physical signals to enhance the security of wireless communications. The diversity and time-varying characteristics of wireless channels across the space-time-frequency domain, combined with the unique and complementary properties of communication channels between transmitters and receivers, provide substantial advantages for PLS technologies. In this context, the ability of RIS to manipulate the propagation environment-enhancing signal strength for legitimate users while reducing the signal strength of illegal users also performs well in improving system security \cite{c12}. RIS introduces additional degrees of freedom for physical layer security enhancement, positioning it as a promising technology for PLS implementation \cite{c13}. Consequently, the integration of PLS and RIS technology for security enhancement has attracted considerable research attention \cite{c14}. However, the limited coverage of conventional RIS constrains its full potential. To address this limitation, integrating RIS with unmanned aerial vehicle (UAV) technology has emerged as a highly promising research direction, leveraging the high mobility of UAVs, inherent line-of-sight (LoS) air-to-ground channels, and on-demand deployment characteristics \cite{c15}. By mounting an RIS on a UAV, referred to as an aerial RIS (ARIS), the resulting UAV-mounted RIS can intelligently reflect signals from the aerial positions \cite{c18}, offering significant advantages over terrestrial RIS. Whereas traditional RIS is limited to reflecting signals within a 180-degree half-space, ARIS enables full 360-degree reflection coverage \cite{c19}. 
 This approach enhances PLS in an energy-efficient manner by optimizing both the orientation and placement of the RIS, thereby increasing network flexibility \cite{c21}.

In addition, the design of RIS-aided transmission in HST communication systems has been investigated in prior studies \cite{c22,c23,c24}. Most of these studies focused on deploying RIS panels at fixed locations, 
which constrains communication performance based on the selected deployment site. Although some research has explored UAV-assisted HST communications \cite{c27,c28}, these studies have not incorporated RIS technology to enhance  communication capabilities further. UAV-RIS-assisted HST communication was investigated in \cite{c30,c31}, but these works did not address the security challenges inherent in HST communications. 

Motivated by the aforementioned challenges, this study investigates ARIS-assisted secure transmission for HST communication systems in harsh and blockage-dominated wireless environments, with a particular focus on short critical railway segments where rapid and on-demand secure coverage enhancement is more practical than dense permanent infrastructure deployment. The weighted sum secrecy rate (WSSR) is employed as the performance metric to evaluate the quality of service. To address this problem, we propose a block coordinate descent (BCD) algorithm to maximize the WSSR while satisfying the constraints on BS transmit power and the unit modulus of the ARIS reflecting coefficients (RC).

The major contributions are summarized as follows:
\begin{itemize}
\item We design a secure transmission method for an ARIS-assisted downlink MU-MISO HST communication system with an eavesdropper (Eve) to enhance HST secure coverage, where the UAV-mounted RIS is quasi-statically deployed over a target railway segment. Different from conventional fixed terrestrial RIS-assisted secure systems, the considered formulation is motivated by localized and on-demand secure coverage enhancement in blockage-prone HST scenarios. Then, we formulate a WSSR optimization problem that is maximized by jointly optimizing the beamforming and the RC. The problem presents significant challenges due to non-concave nature of the objective function and the high coupling among the optimization variables. To solve this problem, we develop a novel linearization technique to transform the objective function into a more tractable form.
\item The BCD method is subsequently employed to iteratively solve the optimization problem. The original optimization problem is decomposed into two subproblems: the beamforming design subproblem and the ARIS phase shift optimization subproblem. To tackle the active beamforming optimization subproblem, the successive convex approximation (SCA) technique is applied. For the phase shift optimization subproblem, it is formulated as a quadratic program (QP) with respect to (w.r.t.) the RIS phase shift. The alternating direction method of multipliers (ADMM) is then utilized to solve this subproblem. 
\item Owing to hardware limitations, the phase shift of the RIS elements are typically quantized, which complicates the acquisition of perfect channel state information (CSI) for the investigated system. Consequently, the analysis is extended to examine the system's performance under conditions involving discrete phase shift and imperfect CSI. The proposed algorithm is directly applicable to scenarios with discrete RC and imperfect CSI.  
\item The performance of the proposed algorithm is evaluated through a set of benchmark schemes, including ideal phase shift, imperfect CSI, discrete phase shift, random beamforming, and random phase shift. These benchmarks are selected to respectively reflect the performance reference bound, practical channel/hardware impairments, and the performance gains brought by the joint optimization of active beamforming and ARIS phase shift. The numerical results verify the necessity and effectiveness of the proposed design for enhancing the secrecy performance of ARIS-assisted HST communication systems.
\end{itemize}


The structure of the paper is outlined as follows. Section \ref{rw} reviews the related literature. Section \ref{sy} presents the system model and formulates the WSSR maximization problem. Section \ref{bfrc} develops the joint design for WSSR maximization. Section \ref{add} extends the investigated system to the cases of discrete phase shift and imperfect CSI. Simulation results are presented in Section \ref{nr} to validate the performance of the proposed algorithm. Finally, Section \ref{con} concludes the paper.

{\textbf{\textit{Notations}}}: In this paper, scalar quantities are represented by italic letters, while boldface lowercase and uppercase letters  denote vectors and matrices, respectively. $\mathbb{C} ^{X\times Y}$ represents the space of $X\times Y$ complex-valued matrices. The conjugate, transpose, and conjugate transpose of $\mathbf{H}$ are denoted as $\mathbf{H}^\dagger $, $\mathbf{H}^{\mathrm{T}}$, $\mathbf{H}^{\mathrm{H}}$, respectively. $\left[ \cdot\right]^+ $ refers to the operation of negative truncation, i.e., $\left[ x\right] ^+ = \max\left(x,0\right)$, $\left[\mathbf{H}\right] _{m,n}$ denotes the element at the $m$th row and $n$th column of matrix $\mathbf{H}$, $\mathbf{H}\succeq \mathbf \mathbf{0}$ suggests that $\mathbf{H}$ is a semi-definite matrix, $\left\lvert \cdot \right\rvert $ is the absolute value of its argument. $\left\lVert \cdot \right\rVert$ is the Euclidean norm of a vector or the Frobenius norm of matrix, $\mathfrak{R} \left\{ \cdot \right\} $ is the real part of its argument. $\mathbb{E} \left\{ \cdot\right\}$ denotes the expectation operator, $\mathrm{diag} \left\{ \cdot\right\}$ represents the diagonal operation. $\otimes$ is the Kronecker product of vectors or matrices. The identity matrix with proper dimension is denoted by $\mathbf{I}$. $\mathcal{C} \mathcal{N} (\mathbf{0},\mathbf{I})$ represents a random vector following a circularly symmetric complex Gaussian (CSCG) distribution, with zero mean   and variance $\mathbf{I}$.

\section{Related Work} \label{rw}
\subsection{RIS-Assisted Secure Communications}
In PLS, the evaluation of the secrecy rate (SR) is vital, as it determines the fundamental boundary of secure communication systems. The SR refers to the gap in achievable rates between the intended receiver and any potential Eve. In recent years, research on RIS has gained attention due to its role in enhancing SR in PLS. Numerous studies have examined the advantages of RIS-aided secure communications, integrating both active and passive beamforming along with various optimization methods to improve the SR \cite{c32,c35,c36}. In \cite{c32}, the authors investigated an RIS-assisted secure MU-MISO communication system with multiple Eves. 
In \cite{c35}, the authors sought to optimize the WSSR and proposed a linearization technique to approximate the non-convex objective function, simplifying its solution. 
In \cite{c36}, the authors explored PLS in RIS-aided cell-free networks, framing the WSSR maximization problem.
However, the aforementioned studies primarily focus on utilizing fixed RIS locations, such as those installed on building exteriors or walls.
 This method is inherently limited, as the cascaded path loss effect leads to significant fluctuations in signal power depending on RIS placement.

\subsection{UAV-Assisted Secure Communications}
In this context, several studies have been conducted, such as \cite{c37,c38,c39,c41}, to evaluate the effectiveness of establishing secure communications between a base station (BS) and a legitimate user using ARIS-based networks in the presence of Eves. In \cite{c37}, the authors focused on optimizing the SR in ARIS networks by optimizing the ARIS selection process. In \cite{c38}, the authors investigated the maximization of the ergodic SR in ARIS networks by optimizing the deployment of the ARIS and RIS components in a three-dimensional (3D) space for end-users. 
In \cite{c39}, the authors introduced a learning-based approach to maximize the SR through the joint optimization of beamforming, UAV placement, AN, and RIS passive beamforming Additionally, the work in \cite{c41} focused on jointly optimizing the beamforming, RIS phase shift, and AN to maximize the SR. However, all of these studies assume static environments for both legitimate users and Eves, neglecting the effects of mobility, particularly in high-mobility scenarios such as HST environments.

\subsection{Secure Communications for HST}
The authors in \cite{c42} emphasized that security is a critical factor in the development of HST systems, and the security performance of  HST wireless communication systems is of paramount importance. 
PLS is regarded as a promising approach to improving the security of wireless communication systems and can be effectively applied to HST communication systems. Several studies in recent years have investigated secure communication for HSTs \cite{c43,c44,c45,c48}. In \cite{c43}, the authors investigated RIS-assisted secure HST communications. Simulation results demonstrated that employing RIS can substantially enhance secrecy performance. In \cite{c44}, the authors leverages the sensing function of integrated sensing and communication technologies to support both eavesdropping detection and PLS techniques. In \cite{c45}, the authors proposed a transmission strategy involving the joint optimization of beamforming and AN. In \cite{c48}, the authors introduced a novel transmission strategy to enhance HST wireless communication security by jointly optimizing beamforming and AN.
However, research on secure wireless communications for HST remains limited. Moreover, none of the aforementioned studies considered integrating UAV and RIS or applying them to PLS techniques to enhance HST communication security. Although WSSR or SR optimization has been studied in conventional RIS-assisted secure communication systems, the existing works mainly focus on fixed RIS deployments or static communication scenarios. In contrast, secure transmission design for ARIS-assisted HST communications remains largely unexplored. The main novelty of this paper is therefore not the mathematical form of the generic optimization problem itself, but its formulation and solution in an ARIS-assisted HST setting, where ARIS deployment, blocked direct links, and high-mobility channel characteristics must be jointly taken into account. Based on this scenario-specific model, we develop a joint beamforming and ARIS phase shift optimization algorithm using the BCD framework, where the beamforming subproblem is handled by SCA and the ARIS phase shift subproblem is addressed by ADMM. We compare our proposed scheme with the relevant works in Table \ref{tab:comparison}.



\begin{table*}[!t]
  \caption{Comparison Among Relevant Works and Our Proposed Scheme}
  \label{tab:comparison}
  \centering
  \renewcommand{\arraystretch}{1.15}
  \setlength{\tabcolsep}{4pt}
  \footnotesize
  \resizebox{\textwidth}{!}{%
  \begin{tabular}{|C{1.3cm}|C{0.6cm}|C{2.8cm}|C{0.6cm}|C{2.2cm}|C{3.1cm}|C{5.0cm}|} 
  \hline
  \rowcolor{gray!25}
  \textbf{Publication} & \textbf{ARIS} & \textbf{ HST Mobility / Mobile User} & \textbf{PLS} & \textbf{Doppler-Aware HST channel} & \textbf{Design Objective} & \textbf{Algorithm} \\
  \hline
  \cite{c32} & $\times$ & $\times$ & $\checkmark$ & $\times$ & WSSR & SCA+ADMM+element-wise BCD+AO \\
  \hline
  \cite{c35} & $\times$ & $\times$ & $\checkmark$ & $\times$ & WSSR &  penalty dual decomposition+Lagrange dual+AO \\
  \hline
  \cite{c36} & $\times$ & $\times$ & $\checkmark$ & $\times$ & WSSR & SDR+SCA+linear conic relaxation+AO \\
  \hline
  \cite{c37} & $\checkmark$ & $\times$ & $\checkmark$ & $\times$ &  average secrecy capacity & improved particle swarm optimization \\
  \hline
  \cite{c38} & $\checkmark$ & $\checkmark$ & $\checkmark$ & $\times$ & secrecy capacity & metaheuristic \\
  \hline
  \cite{c39} & $\checkmark$ & $\times$ & $\times$ & $\times$ & worst-case SR & post-decision state deep Q-network
  combined with Fourier feature mapping\\
  \hline
  \cite{c41} & $\checkmark$ & $\checkmark$ & $\checkmark$  & $\times$ &  average SR & twin delayed deep deterministic policy gradient+deep reinforcement learning \\
  \hline
  \cite{c43} & $\times$ & $\checkmark$ & $\checkmark$ & $\times$ & maximum scheduled flows & maximum ratio transmission+local search+AO \\
  \hline
  \cite{c44} & $\times$ & $\checkmark$ & $\checkmark$ & $\checkmark$ & sum rate & SCA+AO \\
  \hline
  \textbf{This work} & $\checkmark$ & $\checkmark$ & $\checkmark$ & $\checkmark$ & WSSR & SCA+ADMM+BCD \\
  \hline
  \end{tabular}%
  }
\end{table*}

\section{System Model And Problem Formulation} \label{sy}

\subsection{System Model}
As illustrated in F{}ig.~\ref{fig:1}, this study considers a secrecy ARIS-assisted downlink MU-MISO HST communication system. In this scenario, a BS equipped with multiple antennas, assisted by an ARIS, transmits information securely to $K$ single-antenna users in the presence of a malicious single-antenna Eve. An RIS mounted on an UAV facilitates reliable communication. The UAV maintains a fixed position within the coverage area of both the track-side BS and itself, ensuring the communication range includes both entities. 

In this work, the above fixed-position assumption should be understood as a quasi-stationary hovering assumption over a target railway segment during one communication/optimization interval, rather than a long-term invariant deployment over the entire railway line. We assume that the UAV is equipped with basic flight-control and attitude-stabilization mechanisms, so that platform-induced position and attitude fluctuations remain sufficiently small within the considered interval. Therefore, the large-scale geometry between the BS, the ARIS, and the target train segment can be treated as approximately constant for optimization and analysis. 
\begin{remark}
{\textbf{Difference between the quasi-static ARIS and a fixed terrestrial RIS:}}
The quasi-static assumption in this paper means that the UAV-mounted RIS remains approximately stationary only within one short communication/optimization interval over a target railway segment. This assumption is adopted to make the joint active beamforming and ARIS phase shift optimization analytically tractable. From the perspective of one instantaneous optimization block, the ARIS position is indeed fixed and the resulting cascaded channel has a similar algebraic form to that of a conventional fixed RIS-assisted system. However, the considered ARIS should not be interpreted as a permanently installed terrestrial RIS. A terrestrial RIS is usually mounted at a predetermined location, such as on a building wall or along the railway side, and its service region is constrained by the installation site. In contrast, the ARIS can be deployed at a suitable three-dimensional aerial location before or during a target service mission, which is particularly useful for blockage-prone railway segments, emergency communication recovery, and temporary secure reinforcement. Furthermore, the ARIS-user and ARIS-Eve links in the considered HST scenario are affected by high-speed mobility and Doppler shifts. Therefore, although the ARIS is quasi-static within each optimization interval, the proposed model captures a different deployment and propagation scenario from conventional fixed terrestrial RIS-assisted secure communication systems.
\end{remark}

It is also emphasized that the considered communication links are intended to be maintained over a localized service window, where the UAV remains inside the effective coverage region of both the BS and the target train segment, and where no additional blockage occurs on the BS-ARIS link. Hence, the proposed model is particularly suitable for short-duration and segment-level secure coverage enhancement, such as local blockage mitigation, emergency communication recovery, or temporary secure reinforcement over critical railway sections, rather than continuous full-route aerial escorting.

\begin{remark}
{\textbf{Operational duty cycle and service scheduling:}}
The quasi-static assumption in this paper means that the UAV-mounted RIS remains approximately stationary only within one short The proposed ARIS-assisted UAV is intended for localized and on-demand secure coverage enhancement over a target railway segment, rather than continuous full-route service. Therefore, the WSSR maximization problem formulated in this paper should be interpreted as the physical-layer secure transmission design during an active service window, in which the target HST is located inside the effective ARIS-supported segment and the UAV-mounted RIS remains quasi-stationary.
\end{remark}

In practical operation, the active service duration of the UAV for one train depends on the length of the ARIS-supported segment and
the train speed. Let $L_{\rm seg}$ denote the effective supported segment length, $v$ denote the train speed, and $T_{\rm g}$ denote
the preparation and guard time required for UAV activation, link alignment, and service switching. Then, the active service time for
one train can be approximately expressed as 
\begin{equation}
T_{\rm on} \approx \frac{L_{\rm seg}}{v}+T_{\rm g}.
\end{equation}
If the average train headway is $T_{\rm h}$, the corresponding operational duty cycle can be approximated as
\begin{equation}
\rho \approx \min\left\{1,\frac{T_{\rm on}}{T_{\rm h}}\right\}.
\end{equation}
For non-periodic train arrivals, $\rho$ can be evaluated over a given operation horizon according to the union of active service intervals associated with all scheduled train arrivals.

This duty-cycle interpretation indicates that, when $T_{\rm on} \ll T_{\rm h}$, the UAV may remain idle for a large portion of the operation period if it continuously hovers at the service location. Hence, continuous hovering is not necessarily energy-efficient for intermittent HST arrivals. A more practical operation mode is timetable-aware or event-triggered service scheduling, where the UAV is activated before the train enters the target segment, remains quasi-stationary during the active service window, and then switches to standby, landing, recharging, or relocation after the train leaves the segment. When multiple trains arrive with short headways, adjacent service windows can be grouped to reduce repeated takeoff and landing overhead. Therefore, the proposed ARIS-assisted architecture is most suitable for short-term secure coverage enhancement in blockage-prone, emergency, or security-sensitive railway segments. Long-term duty-cycle-aware UAV scheduling, battery management, charging planning, and multi-segment service assignment are beyond the scope of this paper and will be investigated in future work.

The BS is equipped with a uniform planar array (UPA) of $ M =  M_{y_{\mathrm{B}}} \times  M_{z_{\mathrm{B}}}$ antennas, the ARIS consists of a UPA with $N =  N_{x_{\mathrm{R}}} \times  N_{y_{\mathrm{R}}}$ elements. For brevity, we define the sets $\mathcal{K} = \left\{ 1,\cdots ,K \right\} $ and $\mathcal{N} = \left\{ 1,\cdots ,N \right\}$ as the indices of the users and ARIS reflecting elements, respectively. Obstacles such as buildings or vegetation block the direct links between the BS and the users, as well as between the BS and Eve. Consequently, both the users and Eve can only receive signals reflected by the ARIS. We denote $\mathbf{h}_{\mathrm{BR}}\in \mathbb{C} ^{N\times M}$ as the channel between the BS and the ARIS, $\mathbf{h}_{\mathrm{R},k}\in \mathbb{C} ^{N}$ as the channel between the ARIS and the $k$th user, and $\mathbf{h}_{\mathrm{RE}}\in \mathbb{C} ^{N}$ as the channel between the ARIS and the Eve. Since the practical methods for CSI acquisition in RIS-assisted communication systems have been extensively studied in prior work \cite{r3,r4}, we assume perfect CSI knowledge at both the BS and ARIS in the main formulation for analytical tractability. 

\begin{figure}[!t]
  \centering
  {\includegraphics[scale=0.22]{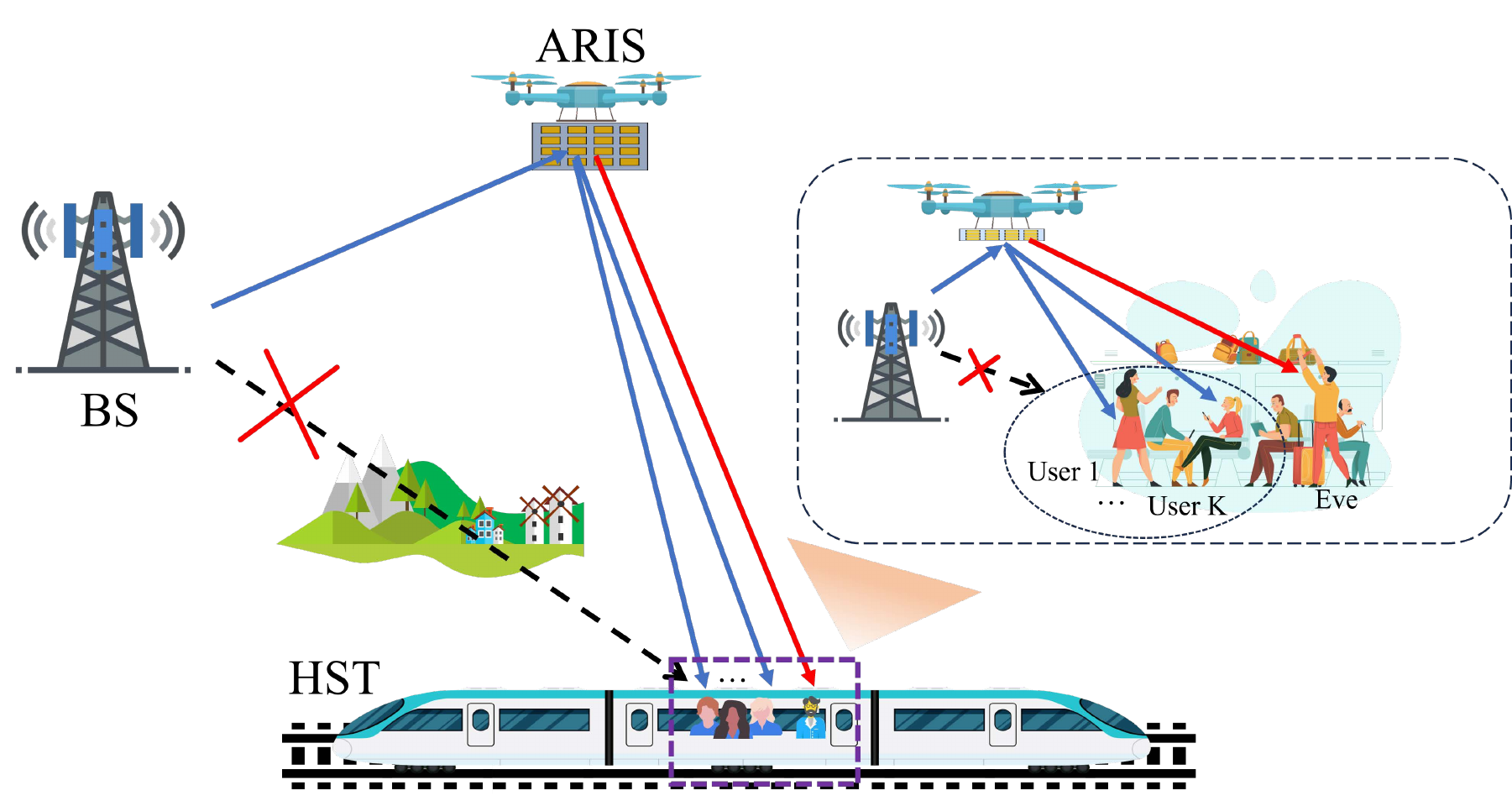}}
  \caption{ \label{fig:1} System model of an ARIS-assisted HST communications.}
\end{figure}

In the 3D Cartesian coordinate system, the positions of the BS, the ARIS, user $k$, and the Eve are denoted as $\boldsymbol{q}_{\mathrm{B}} = \left[\boldsymbol{q}_{\mathrm{B}}^{(xy)}; H_{\mathrm{B}}\right] \in\mathbb{R}^{3\times1}$,  $\boldsymbol{q}_{\mathrm{R}} = \left[\boldsymbol{q}_{\mathrm{R}}^{(xy)}; H_{\mathrm{R}}\right] \in\mathbb{R}^{3\times1}$, $\boldsymbol{q}_{\mathrm{U}} = \left[\boldsymbol{q}_{\mathrm{U}}^{(xy)}; H_{\mathrm{U}}\right] \in\mathbb{R}^{3\times1}$, and $\boldsymbol{q}_{\mathrm{E}} = \left[\boldsymbol{q}_{\mathrm{E}}^{(xy)}; H_{\mathrm{E}}\right] \in\mathbb{R}^{3\times1}$, respectively. Here, $\boldsymbol{q}_{\mathrm{B}}^{(xy)} \in\mathbb{R}^{2\times1}, \boldsymbol{q}_{\mathrm{R}}^{(xy)}\in\mathbb{R}^{2\times1}, \boldsymbol{q}_{\mathrm{U}}^{(xy)}\in\mathbb{R}^{2\times1}$, and $\boldsymbol{q}_{\mathrm{E}}^{(xy)}\in\mathbb{R}^{2\times1}$ are their horizontal positions, while $H_{\mathrm{B}}, H_{\mathrm{R}}, H_{\mathrm{U}}$, and $H_{\mathrm{E}}$ denote their respective heights.

The reflection phase matrix is denoted by $\mathbf{\Theta }=\mathrm{diag}\left\{ \theta _1,\cdots ,\theta _N \right\} \in \mathbb{C} ^{N\times N} $, where $\theta _n=e^{j\varphi _n}$ and $\varphi _n\in \left[ 0,2\pi \right), \forall n \in \mathcal{N}$, $\varphi _n \in\mathcal{F}$ represents the phase shift of the $n$th ARIS element. 
Following \cite{c32}, we assume the amplitude of the ARIS to be normalized due to its passive characteristic, leading to $\mathcal{F}=\left\{\theta_{n}|\theta_{n}=e^{j\varphi_{n}},\varphi_{n}\in[0,2\pi)\right\}$. 

The signal transmitted from the BS to the $k$th user is denoted as $s_k$, where $\mathbb{E}\left\{s_k\right\}=0$ and $\mathbb{E}\left\{|s_k|^2\right\}=1$. Thus, the transmitted signals from the BS can be expressed as 
\begin{equation}
  \mathbf{x}=\sum_{k=1}^K\mathbf{w}_{k}s_k,
\end{equation}
where $\mathbf{w}_k\in \mathbb{C} ^{M}$ is the transmit beamforming vector for the $k$th user. Thus, the received signal at the $k$th user is given by 
\begin{align}
  y_k  &= \left(\mathbf{h}_{\mathrm{R},k}^{\mathrm{H}}\mathbf{\Theta }^{\mathrm{H}}\mathbf{h}_{\mathrm{BR}} \right)\mathbf{x} + n_k \\ 
  & = {\mathbf{h}_{\mathrm{R},k}^{\mathrm{H}}\mathbf{\Theta }^{\mathrm{H}}\mathbf{h}_{\mathrm{BR}}\mathbf{w}_ks_k}  + \sum_{j=1,j\neq k}^K{\mathbf{h}_{\mathrm{R},k}^{\mathrm{H}}\mathbf{\Theta }^{\mathrm{H}}\mathbf{h}_{\mathrm{BR}}\mathbf{w}_js_j}+n_k, \nonumber 
\end{align}
where $\mathbf{h}_{k}^{\mathrm{H}} =\mathbf{h}_{\mathrm{R},k}^{\mathrm{H}}\mathbf{\Theta }^{\mathrm{H}}\mathbf{h}_{\mathrm{BR}}$ is the equivalent channel of the composite BS-ARIS-user link, $n_k$ represents the  additive white Gaussian noise (AWGN) at the $k$th user, with zero mean and variance $\sigma_k^2$. Furthermore, the information of the $k$th user intercepted by the Eve can be given by
  \begin{align}  
  y_{e,k} &=  \left( \mathbf{h}_{\mathrm{RE}}^{\mathrm{H}}\mathbf {\Theta }^{\mathrm{H}}\mathbf{h}_{\mathrm{BR}}   \right) \mathbf{x} + n_{e}   \\
  & = {\mathbf{h}_{\mathrm{RE}}^{\mathrm{H}}\mathbf{\Theta }^{\mathrm{H}}\mathbf{h}_{\mathrm{BR}}\mathbf{w}_ks_k}  + \sum_{j=1,j\neq k}^K{\mathbf{h}_{\mathrm{RE}}^{\mathrm{H}}\mathbf{\Theta }^{\mathrm{H}}\mathbf{h}_{\mathrm{BR}}\mathbf{w}_js_j}+n_{e}, \nonumber
\end{align}
where $\mathbf{h}_{\mathrm{E}}^{\mathrm{H}} = \mathbf{h}_{\mathrm{RE}}^{\mathrm{H}}\mathbf {\Theta }^{\mathrm{H}}\mathbf{h}_{\mathrm{BR}} $ is the equivalent channel of the composite BS-ARIS-Eve link, and $n_{e}$ is the AWGN at Eve with zero mean and variance $\sigma _{e}^{2}$.

By defining the vector $\boldsymbol{\theta }=\left[ \theta _1,\dots ,\theta _N \right] ^{\mathrm{T}}$, the equivalent channels $\mathbf{h}_{k}^{\mathrm{H}} =\mathbf{h}_{\mathrm{R},k}^{\mathrm{H}}\mathbf{\Theta }^{\mathrm{H}}\mathbf{h}_{\mathrm{BR}}$ and $\mathbf{h}_{\mathrm{E}}^{\mathrm{H}} = \mathbf{h}_{\mathrm{RE}}^{\mathrm{H}}\mathbf {\Theta }^{\mathrm{H}}\mathbf{h}_{\mathrm{BR}}$ can be expressed as
\begin{align}
  \mathbf{h}_{k}^{\mathrm{H}}&\triangleq \boldsymbol{\theta }^{\mathrm{H}}\mathbf{H}_k, \forall k \in \mathcal{K}, \\
  \mathbf{h}_{\mathrm{E}}^{\mathrm{H}}&\triangleq \boldsymbol{\theta }^{\mathrm{H}}\mathbf{H}_{\mathrm{E}},
\end{align}
where $\mathbf{H}_k=\mathrm{diag}\left( \mathbf{h}_{\mathrm{R},k}^{\mathrm{H}} \right) \mathbf{h}_{\mathrm{BR}} \in \mathbb{C} ^{N\times M}$ is defined as the cascaded BS-ARIS-user channel, and $\mathbf{H}_{\mathrm{E}}=\mathrm{diag}\left( \mathbf{h}_{\mathrm{RE}}^{\mathrm{H}} \right) \mathbf{h}_{\mathrm{BR}} \in \mathbb{C} ^{N\times M}$ is defined as the cascaded BS-ARIS-Eve channel.

\subsection{Channel Model}
In high-speed mobility scenarios, the Rician fading channel has been widely employed in HST communications, as discussed \cite{r6,r7}. Consequently, we model all the involved channels using the Rician fading model. 

\subsubsection{BS-ARIS Channel}
The channel between the BS and ARIS ${\mathbf{h}}_{\mathrm{BR}} \in \mathbb{C} ^{N\times M}$  can be given by
\begin{equation} 
  \mathbf{h}_{\mathrm{BR}}=\sqrt{L_{\mathrm{BR}}}\left( \sqrt{\frac{\kappa _{\mathrm{BR}}}{\kappa _{\mathrm{BR}}+1}}\overline{\mathbf{h}}_{\mathrm{BR}}+\sqrt{\frac{1}{\kappa _{\mathrm{BR}}+1}}\widetilde{\mathbf{h}}_{\mathrm{BR}} \right),
\end{equation}
where $\kappa _{\mathrm{BR}} \geq 0 $ is the Rician K-factor, $L_{\mathrm{BR}}$ represents the distance-dependent path loss, which can be given as $L_{\mathrm{BR}} = L_0\left( \frac{d_{\mathrm{BR}}}{d_0}\right)^{-\mathrm{\alpha}_\mathrm{BR}}$, where $L_0$ represents the path loss at the reference distance $d_0$ meter (m),  $d_{\mathrm{BR}} = \left\lVert  \boldsymbol{q}_{\mathrm{B}} - \boldsymbol{q}_{\mathrm{R}}  \right\rVert $ is the distance between the BS and ARIS, and $\mathrm{\alpha}_\mathrm{BR}$ denotes the path loss exponent. Additionally, $\overline{\mathbf{h}}_{\mathrm{BR}} \in \mathbb{C} ^{N\times M}$ is the LoS component, given by 
\begin{equation}
  \overline{\mathbf{h}}_{\mathrm{BR}}=\mathbf{a}_{\mathrm{R}}\left( \phi ^{\mathrm{BR}},\delta ^{\mathrm{BR}} \right) \mathbf{a}_{\mathrm{B}}^{\mathrm{H}}\left( \phi ^{\mathrm{BR}},\delta ^{\mathrm{BR}} \right),
\end{equation}
where $\phi ^{\mathrm{BR}}$ and $\delta ^{\mathrm{BR}}$ are the elevation and azimuth arrival angles, respectively, $\mathbf{a}_{\mathrm{B}}\left( \cdot ,\cdot  \right) \in \mathbb{C} ^{M }$ and $\mathbf{a}_{\mathrm{R}}\left( \cdot ,\cdot  \right) \in \mathbb{C} ^{N}$ denote the array response vectors associated with the BS and the ARIS, respectively, which can be given as follows:
\begin{align}
  &\mathbf{a}_{\mathrm{B}}\left( \phi^{\mathrm{BR}} ,\delta^{\mathrm{BR}} \right) = \left[ 1, \cdots ,  e^{\frac{j2\pi^{\mathrm{BR}} d_{y_{\mathrm{B}}}\left( M_{y_{\mathrm{B}}}-1 \right)}{\lambda}\sin \phi^{\mathrm{BR}} \cos \delta^{\mathrm{BR}} } \right] ^{\mathrm{T}}  \nonumber  \\
  &   \qquad \qquad \, \otimes  \left[ 1,\cdots,  e^{\frac{j2\pi d_{z_{\mathrm{B}}}\left( M_{z_{\mathrm{B}}}-1 \right)}{\lambda}\sin \phi^{\mathrm{BR}} \cos \delta^{\mathrm{BR}} } \right] ^{\mathrm{T}} ,
\end{align}

\begin{align}
  & \mathbf{a}_{\mathrm{R}}\left( \phi^{\mathrm{BR}} ,\delta^{\mathrm{BR}} \right) = \left[ 1, \cdots ,  e^{\frac{j2\pi^{\mathrm{BR}} d_{x_{\mathrm{R}}}\left( N_{y_{\mathrm{R}}}-1 \right)}{\lambda}\sin \phi^{\mathrm{BR}} \cos \delta^{\mathrm{BR}} } \right] ^{\mathrm{T}} \nonumber  \\
  &    \qquad \qquad \,  \otimes  \left[ 1,\cdots,  e^{\frac{j2\pi d_{Y_{\mathrm{R}}}\left( N_{y_{\mathrm{R}}}-1 \right)}{\lambda}\sin \phi^{\mathrm{BR}} \cos \delta^{\mathrm{BR}} } \right] ^{\mathrm{T}} ,
\end{align}
where $d_l = \lambda /2, l \in {x, y, z}$ represents the spacing between the elements of the uniform antenna array. Here, $\lambda = \frac{c}{f_c}$ is the wavelength, where $f_c$ and $c$ denote the carrier frequency and the speed of light, respectively. $ N_{y_{\mathrm{B}}} \left( M_{y_{\mathrm{R}}} \right)$ and $ N_{z_{\mathrm{B}}} \left( M_{y_{\mathrm{R}}} \right)$ denote the number of rows and columns, respectively, of the UPA in the 3D plane.

Additionally, $\widetilde{\mathbf{h}}_{\mathrm{BR}}\in \mathbb{C} ^{N \times M}$ is the non-line-of-sight (NLoS) component, with each element modeled as an independent CSCG random vector.

\subsubsection{ARIS-user Channel}
The channel between the ARIS and the $k$th user $\mathbf{h}_{\mathrm{R},k}\in \mathbb{C} ^{N}$ can be given by
\begin{align}
  \mathbf{h}_{\mathrm{R},k}=\sqrt{L_{\mathrm{R},k}}&\left( \sqrt{\frac{\kappa _{\mathrm{R},k}}{\kappa _{\mathrm{R},k}+1}}\overline{\mathbf{h}}_{\mathrm{R},k}+  \sqrt{\frac{1}{\kappa _{\mathrm{R},k}+1}}\widetilde{\mathbf{h}}_{\mathrm{R},k} \right),
\end{align}
where $\kappa _{\mathrm{R},k} \geq 0 $ denotes the Rician K-factor, $L_{\mathrm{R},k}$ represents the path loss. Additionally,  $\overline{\mathbf{h}}_{{\mathrm{R},k}} \in \mathbb{C} ^{N }$ is the LoS component, given by 
\begin{equation}
  \overline{\mathbf{h}}_{\mathrm{R},k}=e^{j2\pi f_d\tau}\mathbf{a}_{\mathrm{R}}\left( \phi ^{\mathrm{R},k},\delta ^{\mathrm{R},k} \right),
\end{equation}
where $f_{d} = v \cos \phi^{\mathrm{R},k} \cos \delta^{\mathrm{R},k} / \lambda$ denotes the Doppler frequency shift, and $\tau$ represents the time slot duration. The term $\widetilde{\mathbf{h}}_{\mathrm{R},k}$ is the NLoS component.

\subsubsection{ARIS-Eve Channel}
The channel between the ARIS and Eve $\mathbf{h}_{\mathrm{RE}}\in \mathbb{C} ^{N}$ can be given by
\begin{equation}
\mathbf{h}_{\mathrm{RE}}=\sqrt{L_{\mathrm{RE}}}\left( \sqrt{\frac{\kappa _{\mathrm{RE}}}{\kappa _{\mathrm{RE}}+1}}\overline{\mathbf{h}}_{\mathrm{RE}}+\sqrt{\frac{1}{\kappa _{\mathrm{RE}}+1}}\widetilde{\mathbf{h}}_{\mathrm{RE}} \right),
\end{equation}
where $\kappa _{\mathrm{RE}} \geq 0 $ is the Rician K-factor, and $L_{\mathrm{RE}}$ represents the distance-dependent path loss. In addition, $\overline{\mathbf{h}}_{{\mathrm{RE}}} \in \mathbb{C} ^{N }$ is the LoS component, given by 
\begin{equation}
  \overline{\mathbf{h}}_{\mathrm{RE}}=e^{j2\pi f_d^{\mathrm{RE}}\tau}\mathbf{a}_{\mathrm{R}}\left( \phi ^{\mathrm{RE}},\delta ^{\mathrm{RE}} \right),
\end{equation}
where $f_{d}^{\mathrm{RE}}=\upsilon \cos \phi ^{\mathrm{RE}} \cos \delta ^{\mathrm{RE}} /\lambda$ represents the Doppler frequency. $\widetilde{\mathbf{h}}_{\mathrm{RE}}$ is the NLoS component.

\begin{remark}
{\textbf{Role of Doppler terms:}}
The Doppler terms in the ARIS-user and ARIS-Eve channels are introduced to characterize the mobility-induced phase evolution
caused by the high-speed movement of the train. It should be noted that, under the ideal assumption of perfect instantaneous CSI and
perfect Doppler tracking, the Doppler-induced phase rotation is included in the instantaneous channel realization and can be compensated in the joint active beamforming and ARIS phase shift design. Therefore, in the perfect CSI case considered in the
main formulation, Doppler does not act as an additional residual impairment in the SINR expressions.
\end{remark}

In this sense, the perfect CSI formulation should be interpreted as a Doppler-aware upper bound benchmark. The role of the Doppler
terms is to reflect the high-mobility HST channel evolution rather than to claim an uncompensated frequency-offset effect. In practical HST communications, imperfect Doppler tracking, channel aging, and outdated CSI may lead to a mismatch between the channel used for optimization and the actual channel during data transmission, which can further affect the achievable secrecy performance. The robust design under such Doppler-induced CSI mismatch is beyond the scope of this paper and will be investigated in future work.


Let $\gamma _k$ and $\gamma _{e,k}$ represent the signal-to-interference-plus-noise ratio (SINR) for the $k$th user and the Eve intercepting the $k$th user's information, respectively, which are given by 
\begin{equation}
  \gamma _k=\frac{\left| \boldsymbol{\theta }^{\mathrm{H}}\bar{\mathbf{H}}_k\mathbf{w}_k \right|^2}{\sum_{j=1,j\ne k}^K{\left| \boldsymbol{\theta }^{\mathrm{H}}\bar{\mathbf{H}}_k\mathbf{w}_j \right|^2}+1},
\end{equation}
\begin{equation}
  \gamma _{e,k}=\frac{\left| \boldsymbol{\theta }^{\mathrm{H}}\bar{\mathbf{H}}_{\mathrm{E}}\mathbf{w}_k \right|^2}{\sum_{j=1,j\ne k}^K{\left| \boldsymbol{\theta }^{\mathrm{H}}\bar{\mathbf{H}}_{\mathrm{E}}\mathbf{w}_j \right|}^2+1},
\end{equation}
where $\bar{\mathbf{H}}_k = {\mathbf{H}}_k \sigma_{k}^{-1}$ and $\bar{\mathbf{H}}_{\mathrm{E}} = {\mathbf{H}}_{\mathrm{E}} \sigma_{e}^{-1}$.

\subsection{Problem Formulation}
In terms of information theory, the core objective of PLS transmission is to optimize the system's SR. The SR is calculated by subtracting the data rate of eavesdropping users from that of legitimate users \cite{c10}. Based on the above channel models, the achievable SR for the $k$th user is
\begin{equation}
  R_{s,k}=\left[ \log _2\left( 1+\gamma _k \right) -\log _2\left( 1+\gamma _{e,k} \right) \right] ^+.
\end{equation}

In this paper, the objective of our design is to maximize WSSR by jointly optimizing the active beamforming vector $\mathbf{w}_k$ and the phase shift $\boldsymbol{\theta }$, while adhering to the constraints on the BS transmit power and the unit modulus at each ARIS element. The optimization problem can be given by
\begin{subequations} \label{eq:P0}
\begin{align} 
  \mathcal{P}_1: \underset{\mathbf{W},\boldsymbol{\theta }}{\max}\, & R_s=\sum_{k=1}^K{\varrho_kR_{s,k}}  \ \label{yyaa1}    \\  
  \mathrm{s}.\mathrm{t}. \,\, & \sum_{k=1}^K{\left\| \mathbf{w}_k \right\| ^2}\le P_T,   \label{yyaa2} \\
  &\left| \theta _n \right|=1,\forall n\in \mathcal{N},  \label{yyaa3}
\end{align}
\end{subequations}
where $\varrho_k \left( 0 \le \varrho_k \le 1 \right)$ is the weight of user $k$, $\mathbf{W}$ is defined as $\mathbf{W}=\left[ \mathbf{w}_{1}^{\mathrm{T}},\cdots ,\mathbf{w}_{K}^{\mathrm{T}} \right]^{\mathrm{T}}$, and $P_T$ is the BS transmit power.

\section{Joint Active Beamforming and RC Design}  \label{bfrc}
The WSSR maximization problem $\mathcal{P}_1$ is challenging to solve because the objective function in \eqref{yyaa1} is non-concave, the variables $\mathbf{W}$ and $\boldsymbol{\theta}$ are coupled, and the unit modulus constraint in \eqref{yyaa3} must be satisfied. Therefore, in this section, we present the BCD algorithm to address the problem $\mathcal{P}_1$. First, we decompose the original problem $\mathcal{P}_1$ into two subproblems: active beamforming optimization and phase shift design. These subproblems are then alternately tackled using the SCA and ADMM methods.

\subsection{Active Beamforming Optimization Subproblem}
Initially, we introduce an SCA approach to optimize the active beamforming while keeping the phase shift variable $\boldsymbol{\theta }$ fixed. The primary idea is to approximate the problem $\mathcal{P}_1$ around the fixed point $\left\{ \mathbf{w}_{k}^{q},\boldsymbol{\theta }^q \right\}$ at the $q$th iteration. The following two lemmas can be used to transform $R_{s,k}$ into a more tractable and solvable expression.
\begin{lemma} \label{le1}
  \cite{c32} For any $a$ and $b$, we have
  \begin{align} 
      \log_2\left( 1+\frac{\left| a \right|^2}{b} \right) & \ge \log_2\left( 1+\frac{\left| \tilde{a} \right|^2}{\tilde{b}} \right) -\frac{\left| \tilde{a} \right|^2}{\tilde{b}\ln 2} \nonumber  \\
      &\quad +\frac{2\Re \{a\tilde{a}\}}{\tilde{b}\ln 2} -\frac{\left| \tilde{a} \right|^2\left( b+\left| a \right|^2 \right)}{\tilde{b}\left( \tilde{b}+\left| \tilde{a} \right|^2 \right) \ln 2}, 
  \end{align}
  \label{eq:le1}
  where $\tilde{a}$ and $\tilde{b}$ are fixed points.
\end{lemma}

Based on Lemma \ref{le1}, the data rate of the $k$th user at the given point $\left\{ \mathbf{w}_{k}^{q},\boldsymbol{\theta }^q \right\} $ can be approximated as
\begin{align} 
  &\log _2\left( 1+\gamma _k \right)  = \log_2\left( 1+\frac{\left| \boldsymbol{\theta }^{\mathrm{H}}\bar{\mathbf{H}}_k\mathbf{w}_k \right|^2}{\sum\nolimits_{j=1,j\ne k}^K{\left| \boldsymbol{\theta }^{\mathrm{H}}\bar{\mathbf{H}}_k\mathbf{w}_j \right|^2}+1}\right) \nonumber \\ 
   & \ge \log_2\left( 1+c_1\left| \left( \boldsymbol{\theta }^q \right) ^{\mathrm{H}}\bar{\mathbf{H}}_k\mathbf{w}_{k}^{q} \right|^2 \right) -\frac{c_1\left| \left( \boldsymbol{\theta }^q \right) ^{\mathrm{H}}\bar{\mathbf{H}}_k\mathbf{w}_{k}^{q} \right|^2}{\ln2}   \nonumber  \\
   &+\frac{c_12\Re \{\left( \mathbf{w}_k \right) ^{\mathrm{H}}\bar{\mathbf{H}}_{k}^{\mathrm{H}}\boldsymbol{\theta }^q\left( \boldsymbol{\theta }^q \right) ^{\mathrm{H}}\bar{\mathbf{H}}_k\mathbf{w}_{k}^{q}\}}{\ln 2} \nonumber \\   
  & -\frac{c_2}{\ln 2}\left( \sum\nolimits_{j=1,j\ne k}^K  {\left| \left( \boldsymbol{\theta }^q \right) ^{\mathrm{H}}\bar{\mathbf{H}}_k\mathbf{w}_j \right|^2} + 1  + \left| \left( \boldsymbol{\theta }^q \right) ^{\mathrm{H}}\bar{\mathbf{H}}_k\mathbf{w}_k \right|^2 \right), 
\end{align}
where
\begin{subequations}
  \begin{align}  
    c_1&=1/\left( \sum\nolimits_{j=1,j\ne k}^K{\left| \left( \boldsymbol{\theta }^q \right) ^{\mathrm{H}}\bar{\mathbf{H}}_k\mathbf{w}_{j}^{q} \right|^2}+1 \right), \tag{\ref{eq:app1}a} \\
    c_2&=  c_1\left| \left( \boldsymbol{\theta }^q \right) ^{\mathrm{H}}\bar{\mathbf{H}}_k\mathbf{w}_{k}^{q} \right|^2/\left( \sum\nolimits_{j=1,j\ne k}^K  \right. {\left| \left( \boldsymbol{\theta }^q \right) ^{\mathrm{H}}\bar{\mathbf{H}}_k\mathbf{w}_{j}^{q} \right|^2} +1  \nonumber \\ 
    &\qquad\qquad\qquad\qquad\quad   \left.  +\left| \left( \boldsymbol{\theta }^q \right) ^{\mathrm{H}}\bar{\mathbf{H}}_k\mathbf{w}_{k}^{q} \right|^2 \right). \tag{\ref{eq:app1}b}
  \end{align}
  \label{eq:app1}
\end{subequations}

Next, we reformulate the information rate corresponding to Eve intercepting the $k$th user. This term is given by
\begin{align} \label{eq:e1}
  -\log _2\left( 1+\gamma _{e,k} \right) &=\log_2\left( 1+\sum_{j=1,j\ne k}^K{\left| \left( \boldsymbol{\theta }^q \right) ^{\mathrm{H}}\bar{\mathbf{H}}_{\mathrm{E}}\mathbf{w}_j \right|^2} \right) \nonumber \\ 
  & \quad -\log_2\left( 1+z_{e,k} \right), 
\end{align}
where $z_{e,k}=\sum\nolimits_{j=1}^K{\left| \left( \boldsymbol{\theta }^q \right) ^{\mathrm{H}}\bar{\mathbf{H}}_{\mathrm{E}}\mathbf{w}_j \right|^2}$. To address  \eqref{eq:e1}, we present the following lemma.

\begin{lemma} 
  \cite{c35} For any $\left\{ a_i \right\} _{i=1}^{K}$, we have
  \begin{align}   \label{le2}
    &\log _2\left( 1+\sum_{i=1}^K{\left| a_i \right|^2} \right)\ge \log _2\left( 1+\sum_{i=1}^K{\left| \tilde{a}_i \right|^2} \right) -\frac{\sum\nolimits_{i=1}^K{\left| \tilde{a}_i \right|^2}}{\ln 2}  \nonumber \\
    &+ \frac{\sum\nolimits_{i=1}^K{2\Re \left\{ \tilde{a}_{i}^{*}a_i \right\}}}{\ln 2}-\frac{\left( \sum\nolimits_{i=1}^K{\left| \tilde{a}_i \right|^2} \right) \left( 1+\sum\nolimits_{i=1}^K{\left| a_i \right|^2} \right)}{\left( 1+\sum\nolimits_{i=1}^K{\left| \tilde{a}_i \right|^2} \right) \ln 2},
  \end{align}
  where $\left\{ \tilde{a}_i \right\} _{i=1}^{K}$ is fixed point.
\end{lemma}
\begin{proof}

By using Lemma \ref{le1} with $b=\tilde{b}=1$, we obtain
\begin{align} 
  \log _2\left( 1+\left| a \right|^2 \right) &\ge \log _2\left( 1+\left| \tilde{a} \right|^2 \right) -\frac{\left| \tilde{a} \right|^2}{\ln 2}   \nonumber \\ 
  & +\left. \frac{2\Re \left\{ \tilde{a}^*a \right\}}{\ln 2}-\frac{\left| \tilde{a} \right|^2\left( 1+\left| a \right|^2 \right)}{\left( 1+\left| \tilde{a} \right|^2 \right) \ln 2} \right. .
\end{align}

Then, by fixing $a_i$ for $i=2,\dots ,K$, it follows
\begin{align} \label{eq:e2}
  \log _2\left( S+\left| a_1 \right|^2 \right) &\ge \log _2\left( S+\left| \tilde{a_1} \right|^2 \right) -\frac{\left| \tilde{a_1} \right|^2}{\ln 2}  + \frac{2\Re \left\{ \tilde{a_1}^*a_1 \right\}}{\ln 2}\nonumber  \\
  & \quad - \frac{\left| \tilde{a_1} \right|^2\left( S+\left| a_1 \right|^2 \right)}{\left( S+\left| \tilde{a_1} \right|^2 \right) \ln 2},
\end{align}
where $S=1+\sum\nolimits_{i=2}^K{\left| a_i \right|}^2$. Thus, by using \eqref{eq:e2} for $a_i$ with fixed other $a_i$ from $i = 2$ to $i = K$, we have \eqref{le2}.
\end{proof}

According to \eqref{le2}, the first term of \eqref{eq:e1} can be approximated as 
\begin{align}  \label{eq:e3}
  &\log_2\left( 1+\sum_{j=1,j\ne k}^K{\left| \left( \boldsymbol{\theta }^q \right) ^{\mathrm{H}}\bar{\mathbf{H}}_{\mathrm{E}}\mathbf{w}_j \right|^2} \right)   \nonumber \\
  & \ge \log _2\left( 1+\sum_{j=1,j\ne k}^K{\left| \left( \boldsymbol{\theta }^q \right) ^{\mathrm{H}}\bar{\mathbf{H}}_{\mathrm{E}}\mathbf{w}_{j}^{q} \right|^2} \right)  \nonumber  \\
  &-\frac{\sum\nolimits_{j=1,j\ne k}^K{\left| \left( \boldsymbol{\theta }^q \right) ^{\mathrm{H}}\bar{\mathbf{H}}_{\mathrm{E}}\mathbf{w}_{j}^{q} \right|^2}}{\ln 2}  \nonumber \\
  & +\frac{\sum\nolimits_{j=1,j\ne k}^K{2\Re \left\{ \mathbf{w}_{j}^{\mathrm{H}}\bar{\mathbf{H}}_{\mathrm{E}}^{\mathrm{H}}\boldsymbol{\theta }^q\left( \boldsymbol{\theta }^q \right) ^{\mathrm{H}}\bar{\mathbf{H}}_{\mathrm{E}}\mathbf{w}_{j}^{q} \right\}}}{\ln 2}
  \nonumber \\ 
  & -\frac{c_3}{\left( 1+c_3 \right) \ln 2}\left( 1+\sum_{j=1,j\ne k}^K{\left| \left( \boldsymbol{\theta }^q \right) ^{\mathrm{H}}\bar{\mathbf{H}}_{\mathrm{E}}\mathbf{w}_j \right|^2} \right),
\end{align}
where $c_3=\sum\nolimits_{j=1,j\ne k}^K{\left| \left( \boldsymbol{\theta }^q \right) ^{\mathrm{H}}\bar{\mathbf{H}}_{\mathrm{E}}\mathbf{w}_{j}^{q} \right|^2}$.

For brevity, let $\Lambda _{e,k}\Lambda _{e,k}^{\mathrm{H}}=\bar{\mathbf{H}}_{\mathrm{E}}^{\mathrm{H}}\boldsymbol{\theta }^q\left( \boldsymbol{\theta }^q \right) ^{\mathrm{H}}\bar{\mathbf{H}}_{\mathrm{E}}$, \eqref{eq:e3} can be rewritten as
  \begin{align}  
  &\log_2\left( 1+\mathbf{W}^{\mathrm{H}}\Lambda _{e,k}\Lambda _{e,k}^{\mathrm{H}}\mathbf{W} \right)\ge \log _2\left( 1+\left\| \Lambda _{e,k}^{\mathrm{H}}\mathbf{W}^q \right\| ^2 \right) \nonumber \\
  &-\frac{\left\| \Lambda _{e,k}^{\mathrm{H}}\mathbf{W}^q \right\| ^2}{\ln 2}   +\frac{2\Re \left\{ \mathbf{W}^{\mathrm{H}}\Lambda _{e,k}\Lambda _{e,k}^{\mathrm{H}}\mathbf{W}^q \right\}}{\ln 2}  \nonumber \\
  & -\frac{\left\| \Lambda _{e,k}^{\mathrm{H}}\mathbf{W}^q \right\| ^2\left( 1+\left\| \Lambda _{e,k}^{\mathrm{H}}\mathbf{W}\right\| ^2 \right)}{\left( 1+\left\| \Lambda _{e,k}^{\mathrm{H}}\mathbf{W}^q \right\| ^2 \right) \ln 2},
\end{align}
where 
\begin{equation} \label{eq:e4}
  \Lambda _{e,k}\Lambda _{e,k}^{\mathrm{H}}=\mathrm{diag}\left[ \xi _{e,k},\cdots ,\underset{k-\mathrm{th} ~\mathrm{term}}{\underbrace{\mathbf{0}_{M\times M}}},\cdots ,\xi _{e,k} \right].
\end{equation}
In \eqref{eq:e4}, $\xi _{e,k}$ is calculated as $\xi _{e,k}=\bar{\mathbf{H}}_{\mathrm{E}}^{\mathrm{H}}\boldsymbol{\theta}^q\left( \boldsymbol{\theta }^q \right) ^{\mathrm{H}}\bar{\mathbf{H}}_{\mathrm{E}}$. 

Next, we focus on the second term on the right side of \eqref{eq:e1}. Due to the concavity of the logarithm function $\log_2(u)\leq\log_2(u_0)+(u/u_0)-1$, we have \cite{r11,r12}
\begin{equation}
  -\log_2\left( 1+z_{e,k} \right) \ge -\log_2\left( 1+\tilde{z}_{e,k} \right) -\frac{1+z_{e,k}}{\left( 1+\tilde{z}_{e,k} \right) \ln 2}+1,
\end{equation}
where $\tilde{z}_{e,k}=\sum\nolimits_{j=1}^K{\left| \left( \boldsymbol{\theta }^q \right) ^{\mathrm{H}}\bar{\mathbf{H}}_{\mathrm{E}}\mathbf{w}_{j}^{q} \right|^2}$.

Therefore, by disregarding the constant terms, we derive the following problem w.r.t. to $\mathbf{W}_k$:
\begin{subequations}
\begin{align} 
&\mathcal{P}_2:\,\,  \underset{\mathbf{W}}{\min}\,\, \sum_{k=1}^K{\varrho _k \left\{  -\frac{c_12\Re \{\mathbf{w}_{k}^{\mathrm{H}}\bar{\mathbf{H}}_{k}^{\mathrm{H}}\boldsymbol{\theta }^q\left( \boldsymbol{\theta }^q \right) ^{\mathrm{H}}\bar{\mathbf{H}}_k\mathbf{w}_{k}^{q}\}}{\ln 2} \right.}   \nonumber \\
   & \,\,+\frac{c_2\left( \sum\nolimits_{j=1,j\ne k}^K {\left| \left( \boldsymbol{\theta }^q \right) ^{\mathrm{H}}\bar{\mathbf{H}}_k\mathbf{w}_j \right|^2}+1 +\left| \left( \boldsymbol{\theta }^q \right) ^{\mathrm{H}}\bar{\mathbf{H}}_k\mathbf{w}_k \right|^2 \right)}  {\ln 2}  \nonumber \\  
   & \,\,+\frac{z_{e,k}}{\left( 1+\tilde{z}_{e,k} \right) \ln 2}-\frac{2\Re \left\{ \mathbf{W}^{\mathrm{H}}\Lambda _{e,k}\Lambda _{e,k}^{\mathrm{H}}\mathbf{W}^q \right\}}{\ln 2}  \nonumber  \\ 
   & \,\,\left. +\frac{\left\| \Lambda _{e,k}^{\mathrm{H}}\mathbf{W}^q \right\| ^2\left\| \Lambda _{e,k}^{\mathrm{H}}\mathbf{W} \right\| ^2}{\left( 1+\left\| \Lambda _{e,k}^{\mathrm{H}}\mathbf{W}^q \right\| ^2 \right) \ln 2} \right\}
  \\
& \,\, \mathrm{s}. \mathrm{t}. \,\, \sum_{k=1}^K{\left\| \mathbf{w}_k \right\| ^2}\le P_T.  \label{P2b}
\end{align}
\label{eq:P2}
\end{subequations}

It is important to highlight that $\mathcal{P}_2$ is a convex problem, and it can be efficiently addressed using the CVX toolbox.

\subsection{Phase Shift Optimization Subproblem}
In the following analysis, by fixing the beamforming vector $\mathbf{W}$, the phase shift optimization subproblem is considered. The primary challenge in solving the subproblem w.r.t. $\boldsymbol{\theta}$ stems from the  unit modulus constraint. Specifically, near the point $\left\{ \mathbf{w}_{k}^{q},\boldsymbol{\theta }^q \right\}$, the subproblem approximation related to $\boldsymbol{\theta}$ can be expressed as follows:
\begin{subequations}
  \begin{align} 
    \mathcal{P}_3:\,\,& \underset{\boldsymbol{\theta }}{\min}\,\, \sum_{k=1}^K{\varrho _k\left\{ -\frac{c_12\Re \{\boldsymbol{\theta }^{\mathrm{H}}\bar{\mathbf{H}}_k\mathbf{w}_{k}^{q}\left( \mathbf{w}_{k}^{q} \right) ^{\mathrm{H}}\bar{\mathbf{H}}_{k}^{\mathrm{H}}\boldsymbol{\theta }^q\}}{\ln 2} \right.} \nonumber    \\
    &+\frac{c_2 \left( \sum\nolimits_{j=1,j\ne k}^K {\left| \boldsymbol{\theta }^{\mathrm{H}}\bar{\mathbf{H}}_k\mathbf{w}_{j}^{q} \right|^2}+1  +\left| \boldsymbol{\theta }^{\mathrm{H}}\bar{\mathbf{H}}_k\mathbf{w}_{k}^{q} \right|^2 \right)}{\ln 2} \nonumber \\
    & \left. +\frac{\sum\nolimits_{j=1}^K{\left| \boldsymbol{\theta }^{\mathrm{H}}\bar{\mathbf{H}}_{\mathrm{E}}\mathbf{w}_{j}^{q} \right|^2}}{\left( 1+\tilde{z}_{e,k} \right) \ln 2} \right. \nonumber \\ \label{P3a}
    & \left.-\log_2\left( 1+\sum_{j=1,j\ne k}^K{\left| \boldsymbol{\theta }^{\mathrm{H}}\bar{\mathbf{H}}_{\mathrm{E}}\mathbf{w}_{j}^{q} \right|^2} \right) \right\}  \\ 
    & \mathrm{s}.\mathrm{t}. \,\, \left| \theta _n \right|=1,\forall n\in \mathcal{N}.
  \end{align}
  \label{eq:P3}
\end{subequations}
For $\log_2\left( 1+\sum_{j=1,j\ne k}^K{\left| \boldsymbol{\theta }^{\mathrm{H}}\bar{\mathbf{H}}_{\mathrm{E}}\mathbf{w}_{j}^{q} \right|^2} \right) $ in \eqref{P3a}, it can be approximated as
  \begin{align}
  &\log_2\left( 1+\boldsymbol{\theta }^{\mathrm{H}}\Upsilon _{e,k}\Upsilon _{e,k}^{\mathrm{H}}\boldsymbol{\theta } \right) \ge \log _2\left( 1+\left\| \Upsilon _{e,k}^{\mathrm{H}}\boldsymbol{\theta }^q \right\| ^2 \right) \nonumber \\ 
  & -\frac{\left\| \Upsilon _{e,k}^{\mathrm{H}}\boldsymbol{\theta }^q \right\| ^2}{\ln 2}  +\frac{2\Re \left\{ \boldsymbol{\theta }^{\mathrm{H}}\Upsilon _{e,k}\Upsilon _{e,k}^{\mathrm{H}}\boldsymbol{\theta }^q \right\}}{\ln 2} \nonumber \\ 
  & -\frac{\left\| \Upsilon _{e,k}^{\mathrm{H}}\boldsymbol{\theta }^q \right\| ^2\left( 1+\left\| \Upsilon _{e,k}^{\mathrm{H}}\boldsymbol{\theta } \right\| ^2 \right)}{\left( 1+\left\| \Upsilon _{e,k}^{\mathrm{H}}\boldsymbol{\theta }^q \right\| ^2 \right) \ln 2}. 
\end{align}
where $\Upsilon _{e,k}\Upsilon _{e,k}^{\mathrm{H}}=\sum_{j=1,j\ne k}^K{\bar{\mathbf{H}}_{\mathrm{E}}\mathbf{w}_{j}^{q}\left( \mathbf{w}_{j}^{q} \right) ^{\mathrm{H}}\bar{\mathbf{H}}_{\mathrm{E}}^{\mathrm{H}}}$. By neglecting the constant terms, $\mathcal{P}_3$ can be rewritten as 
\begin{subequations}
  \begin{align} 
    \mathcal{P}_4: \,\, &  \underset{\boldsymbol{\theta }}{\min}\,\, \sum_{k=1}^K{\varrho _k\left\{ -\frac{c_12\Re \{\boldsymbol{\theta }^{\mathrm{H}}\bar{\mathbf{H}}_k\mathbf{w}_{k}^{q}\left( \mathbf{w}_{k}^{q} \right) ^{\mathrm{H}}\bar{\mathbf{H}}_{k}^{\mathrm{H}}\boldsymbol{\theta }^q\}}{\ln 2} \right.} \nonumber \\
    & +\frac{c_2\left( \sum\nolimits_{j=1,j\ne k}^K  {\left| \boldsymbol{\theta }^{\mathrm{H}}\bar{\mathbf{H}}_k\mathbf{w}_{j}^{q} \right|^2}+1  +\left| \boldsymbol{\theta }^{\mathrm{H}}\bar{\mathbf{H}}_k\mathbf{w}_{k}^{q} \right|^2 \right) }{\ln 2}\nonumber \\
    & +\frac{\sum\nolimits_{j=1}^K{\left| \boldsymbol{\theta }^{\mathrm{H}}\bar{\mathbf{H}}_{\mathrm{E}}\mathbf{w}_{j}^{q} \right|^2}}{\left( 1+\tilde{z}_{e,k} \right) \ln 2}-\frac{2\Re \left\{ \boldsymbol{\theta }^{\mathrm{H}}\Upsilon _{e,k}\Upsilon _{e,k}^{\mathrm{H}}\boldsymbol{\theta }^q \right\}}{\ln 2}  \nonumber    \\
    & \left. +\frac{\left\| \Upsilon _{e,k}^{\mathrm{H}}\boldsymbol{\theta }^q \right\| ^2\left\| \Upsilon _{e,k}^{\mathrm{H}}\boldsymbol{\theta } \right\| ^2}{\left( 1+\left\| \Upsilon _{e,k}^{\mathrm{H}}\boldsymbol{\theta }^q \right\| ^2 \right) \ln 2} \right\}  \\
    & \mathrm{s}.\mathrm{t}. \,\, \left| \theta _n \right|=1,\forall n\in \mathcal{N}. \label{P4b}
  \end{align}
  \label{eq:P4}
\end{subequations}

Let $\boldsymbol{\varpi }=\sum_{j=1}^K{\mathbf{w}_{j}^{q}\left( \mathbf{w}_{j}^{q} \right) ^{\mathrm{H}}}$,
and then $\sum\nolimits_{j=1,j\ne k}^K{\left| \boldsymbol{\theta }^{\mathrm{H}}\bar{\mathbf{H}}_k\mathbf{w}_{j}^{q} \right|^2}+1+\left| \boldsymbol{\theta }^{\mathrm{H}}\bar{\mathbf{H}}_k\mathbf{w}_{k}^{q} \right|^2=\boldsymbol{\theta }^{\mathrm{H}}\bar{\mathbf{H}}_k\boldsymbol{\varpi }\bar{\mathbf{H}}_{k}^{\mathrm{H}}\boldsymbol{\theta }+1$, $\sum\nolimits_{j=1}^K{\left| \boldsymbol{\theta }^{\mathrm{H}}\bar{\mathbf{H}}_{\mathrm{E}}\mathbf{w}_{j}^{q} \right|^2}=\boldsymbol{\theta }^{\mathrm{H}}\bar{\mathbf{H}}_{\mathrm{E}}\boldsymbol{\varpi }\bar{\mathbf{H}}_{\mathrm{E}}^{\mathrm{H}}\boldsymbol{\theta }$. Then, we have
\begin{subequations}
\begin{align}
  \mathbf{A}_k&=\frac{c_2\left( \bar{\mathbf{H}}_k\boldsymbol{\varpi }\bar{\mathbf{H}}_{k}^{\mathrm{H}} \right)}{\ln 2}+\frac{\bar{\mathbf{H}}_{\mathrm{E}}\boldsymbol{\varpi }\bar{\mathbf{H}}_{\mathrm{E}}^{\mathrm{H}}}{\left( 1+\tilde{z}_{e,k} \right) \ln 2} \frac{\Upsilon _{e,k}\Upsilon _{e,k}^{\mathrm{H}}\left\| \Upsilon _{e,k}^{\mathrm{H}}\boldsymbol{\theta }^q \right\| ^2}{\left( 1+\left\| \Upsilon _{e,k}^{\mathrm{H}}\boldsymbol{\theta }^q \right\| ^2 \right) \ln 2},
\end{align}

\begin{equation}
  \mathbf{v}_k=\frac{\Upsilon _{e,k}\Upsilon _{e,k}^{\mathrm{H}}\boldsymbol{\theta }^q}{\ln 2}+\frac{c_1\left( \bar{\mathbf{H}}_k\mathbf{w}_{k}^{q}\left( \mathbf{w}_{k}^{q} \right) ^{\mathrm{H}}\bar{\mathbf{H}}_{k}^{\mathrm{H}}\boldsymbol{\theta }^q \right) \,\,}{\ln 2}.
\end{equation}
\end{subequations}

Hence, we can rewrite the problem as follows:
\begin{subequations}
\begin{align}
  \mathcal{P}_5: \underset{\boldsymbol{\theta }}{\min}\,\,& \boldsymbol{\theta }^{\mathrm{H}}\mathbf{A}\boldsymbol{\theta }-2\Re \left\{ \boldsymbol{\theta }^{\mathrm{H}}\mathbf{v} \right\} \\
  \mathrm{s}.\mathrm{t}. \,\, & \left| \theta _n \right|=1,\forall n\in \mathcal{N}, \label{P5b}
\end{align}
\label{eq:P5}
\end{subequations}
where $\mathbf{A}=\,\,\sum_{k=1}^K{\varrho _k\mathbf{A}_k}$, and $\mathbf{v}=\,\,\sum_{k=1}^K{\varrho _k\mathbf{v}_k}$.

The simplified subproblem $\mathcal{P}_5$ in \eqref{eq:P5} is a non-convex unit-modulus quadratic program due to the constraint in \cite{r13new}. In principle, several optimization techniques can be employed to handle this type of problem, such as semidefinite relaxation (SDR), element-wise BCD, projected gradient methods, Riemannian manifold optimization, majorization-minimization (MM), and penalty-based methods. However, these methods have different limitations in the considered ARIS phase shift design. For example, SDR requires matrix lifting and usually relies on Gaussian randomization to recover a feasible unit-modulus solution, which may lead to high computational complexity for large-scale ARIS arrays. Element-wise BCD updates the reflecting coefficients sequentially and may suffer from slow convergence when the number of ARIS elements is large. Projected gradient and manifold-based methods can handle the unit-modulus constraint, but they generally require step-size selection, line search, or retraction operations, which increases implementation complexity within the outer BCD framework. To clarify the motivation for using ADMM in this work, a qualitative comparison of representative candidate solvers is provided in Table~\ref{tab_solver_comparison}.

\begin{table*}[!t]  
  \centering
  \caption{Qualitative comparison of candidate solvers for the unit-modulus QP}
  \label{tab_solver_comparison}
  \footnotesize
  \renewcommand{\arraystretch}{1.15}
  \begin{tabular}{p{0.35\columnwidth} p{0.65\columnwidth} p{0.7\columnwidth}}
  \hline
  \textbf{Method} & \textbf{Main consideration} & \textbf{Reason for not being adopted} \\
  \hline
  SDR & 
  Transforms the problem into a lifted semidefinite program. & 
  High computational complexity; Gaussian randomization is usually needed to recover a feasible unit-modulus solution. \\
  \hline
  Element-wise BCD & 
  Updates one reflecting coefficient at a time. & 
  Sequential updates may become inefficient when the number of ARIS elements is large. \\
  \hline
  Projected gradient & 
  Performs gradient update followed by projection onto the unit-modulus set. & 
  Performance is sensitive to step-size selection and may require many iterations. \\
  \hline
  Riemannian optimization & 
  Optimizes directly on the complex unit-modulus manifold. & 
  Usually requires line search or retraction operations, which increases implementation complexity. \\
  \hline
  MM / penalty methods & 
  Solves a sequence of surrogate or penalized subproblems. & 
  Penalty or surrogate design may affect convergence behavior and implementation efficiency. \\
  \hline
  ADMM & 
  Decouples the quadratic objective and the unit-modulus constraint by an auxiliary variable. & 
  Adopted in this work because it yields closed-form updates, preserves unit-modulus feasibility, and fits well into the outer BCD framework. \\
  \hline
  \end{tabular}
\end{table*}

In this work, ADMM is adopted not because it is universally superior to all alternative solvers, but because it well matches the structure of $\mathcal{P}_5$.
To adapt \eqref{eq:P5} for the ADMM framework, the slack variable $\boldsymbol{r}\in \mathbb{C} ^N$ is introduced, and the problem $\mathcal{P}_5$ can be rewritten as
\begin{subequations}
  \begin{align}
    \mathcal{P}_6:\underset{\boldsymbol{r},\boldsymbol{\theta }}{\min}\,\,& \boldsymbol{r}^{\mathrm{H}}\mathbf{A}\boldsymbol{r}-2\Re \left\{ \boldsymbol{r}^{\mathrm{H}}\mathbf{v} \right\}  \\
    \mathrm{s}.\mathrm{t}. \,\, & \left| \theta _n \right|=1,\forall n\in \mathcal{N},  \label{P6b} \\ 
       & \boldsymbol{r}=\boldsymbol{\theta }. \label{P6c}
  \end{align}
  \label{eq:P6}
\end{subequations}

The augmented Lagrangian equation of $\mathcal{P}_6$ can be given as 
\begin{align}
  \mathcal{L} \left( \boldsymbol{r},\boldsymbol{\theta },\boldsymbol{p} \right) &=\boldsymbol{r}^{\mathrm{H}}\mathbf{A}\boldsymbol{r}-2\Re \left\{ \boldsymbol{r}^{\mathrm{H}}\mathbf{v} \right\} \nonumber   \\
  & \quad -\Re \left\{ \boldsymbol{p}^{\mathrm{H}}\left( \boldsymbol{r}-\boldsymbol{\theta } \right) \right\} +\frac{\epsilon }{2}\left\| \boldsymbol{r}-\boldsymbol{\theta } \right\| ^2,
\end{align}
where $\boldsymbol{p}\in \mathbb{C} ^N$ is the Lagrange multiplier corresponding to the constraint \eqref{P6c}, and $\epsilon  >  0$ is the penalty parameter.

Let $\left(\boldsymbol{r}^0,{\boldsymbol{\theta}}^0,\boldsymbol{p}^0\right)$ denote the initial primal-dual variables. We use the superscript $l$ on each variable to represent the iteration index. The standard ADMM algorithm proceeds with the following three steps:    
\begin{subequations}
  \begin{align}
    \boldsymbol{r}^{l+1}&=\mathrm{arg} \min_{\boldsymbol{r}} \,\,\mathcal{L} \left( \boldsymbol{r}^l,\boldsymbol{\theta }^l,\boldsymbol{p}^l \right),  \label{eq:thetaa} \\
    \boldsymbol{\theta }^{l+1}&=\mathrm{arg} \min_{\left| \theta _n \right|=1,\forall n\in \mathcal{N}} \,\,\mathcal{L} \left( \boldsymbol{r}^{l+1},\boldsymbol{\theta }^l,\boldsymbol{p}^l \right),  \label{eq:thetab} \\ 
    \boldsymbol{p}^{l+1}&=\boldsymbol{p}^l-\epsilon  \left( \boldsymbol{r}^{l+1}-\boldsymbol{\theta }^{l+1} \right). \label{eq:thetac}
  \end{align}
\end{subequations}

\subsubsection{Optimization $\boldsymbol{r}$}
Here, the first-order optimization solves the gradient $\partial \mathcal{L} /\partial \boldsymbol{r}$ as follows:
\begin{equation} \label{eq:p}
  \frac{\partial \mathcal{L}}{\partial \boldsymbol{r}}=2\mathbf{A}\boldsymbol{r}^{l+1}-2\mathbf{v}-\boldsymbol{p}^l+\epsilon  \left( \boldsymbol{r}^{l+1}-\boldsymbol{\theta }^{l} \right).
\end{equation}

Let \eqref{eq:p} be equal to zero, we have
\begin{equation}  \label{eq:r}
\boldsymbol{r}^{l+1}=\left( 2\mathbf{A}+\epsilon  \mathbf{I} \right) ^{-1}\left( 2\mathbf{v}+\boldsymbol{p}^l+\epsilon  \boldsymbol{\theta }^l \right).
\end{equation}

\subsubsection{Optimization $\boldsymbol{\theta }$}
The solution of \eqref{eq:thetab} is equivalent to
\begin{equation}
\min_{\left| \theta _n \right|=1,\forall n\in \mathcal{N}} \,\,\left\| \boldsymbol{\theta }-\left( \boldsymbol{r}^{l+1}-\epsilon  ^{-1}\boldsymbol{p}^l \right) \right\|,
\end{equation}
which has a closed-form solution
\begin{equation}\label{eq:31}
  \begin{split}
  {\left[ \boldsymbol{\theta }^{l+1} \right]_{n}} = \left\{ 
    \begin{array}{l}
   \frac{\left[ \boldsymbol{r}^{l+1}-\epsilon  ^{-1}\boldsymbol{p}^l \right]_n}{\left| \left[ \boldsymbol{r}^{l+1}-\epsilon  ^{-1}\boldsymbol{p}^l \right] _n \right|}, \, {\left[ \boldsymbol{r}^{l+1}-\epsilon  ^{-1}\boldsymbol{p}^l \right] _n} \ne 0,  \\  \\
   {\left[ \boldsymbol{\theta }^l \right]_{ n}}, \quad \quad \quad \thinspace\thinspace\thinspace~ {\left[ \boldsymbol{r}^{l+1}-\epsilon  ^{-1}\boldsymbol{p}^l \right]_{{\kern 1pt} n}} = 0,{\rm{ }}
  \end{array} \right.
  \end{split}
\end{equation}
where $\left[ \boldsymbol{\theta } \right] _n$ denotes the $n$th entry of $\boldsymbol{\theta }$.

\subsubsection{Optimization $\boldsymbol{p}$}
From \eqref{eq:p}, we can solve \eqref{eq:thetac} as
\begin{equation}
  \boldsymbol{p}^{l+1}=2\mathbf{A}\boldsymbol{r}^{l+1}-2\mathbf{v}.
\end{equation}

The three aforementioned subproblems can be resolved using closed-form solutions. To ensure the convergence of the ADMM algorithm, the following lemma is provided \cite{c32}.
\begin{lemma}
  The ADMM algorithm is guaranteed to converge, regardless of whether $\boldsymbol{\theta }$ is attached to the continuous set, provided that the penalty factor $\epsilon $ satisfies:
  \begin{equation}
    \frac{\epsilon }{2}\mathbf{I}-\mathbf{A}>\mathbf{0}.
  \end{equation}
\end{lemma}
\begin{proof}
 The certification process can be referred to \cite{r13new}.
\end{proof}

The ADMM algorithm is summarized in Algorithm \ref{ADMM}.
\begin{algorithm}[!t]
  \caption{The ADMM Algorithm for Problem $\mathcal{P}_6$}
  \label{ADMM}
  \begin{algorithmic}[1]
  \REQUIRE
  $l = 0$, set a feasible point $\left\{\boldsymbol{r}^0,{\boldsymbol{\theta}}^0,\boldsymbol{p}^0\right\}$, set the accuracy $\varepsilon _R$, and the penalty factor $\epsilon = \ell\left|\left|\mathbf{A}\right|\right|$, where $\ell$ is the minimum integer which satisfies $\epsilon\mathbf{I}/2-\mathbf{A} > 0$.
  \ENSURE
  $\left\{\boldsymbol{r}^{\ast},{\boldsymbol{\theta}}^{\ast},\boldsymbol{p}^{\ast}\right\}$
  \REPEAT
  \STATE Calculate $\boldsymbol{r}^{l+1}$ via \eqref{eq:r}; \\
  Calculate $\boldsymbol{\theta}^{l+1}$ via \eqref{eq:31}; \\
  Calculate $\boldsymbol{p}^{l+1}=2\mathbf{A}\boldsymbol{r}^{l+1}-2\mathbf{v}$; \\
  $l \leftarrow l + 1$;
  \UNTIL $R_s^{l} - R_s^{l-1} < \varepsilon _R$;
  \end{algorithmic}
\end{algorithm}

\subsection{Overall Algorithm and Convergence Analysis}
\begin{algorithm}[!t]
  \caption{The BCD Algorithm for Problem $\mathcal{P}_1$}
  \label{BCD}
  \begin{algorithmic}[1]
  \REQUIRE
  $q = 1$, set $P_T$, $M$, $N$, $\sigma _k$, $\sigma _{e,k}$, $\mathbf{h}_{\mathrm{BR}}$, $\mathbf{h}_{\mathrm{R},k}$, $\mathbf{h}_{\mathrm{RE}}$, and $\varepsilon _R$.
  \ENSURE
    {$\left\{ \mathbf{w}_{k}^{\ast},\boldsymbol{\theta }^{\ast} \right\}$ }
  \REPEAT
  \STATE Obtain $\mathbf{w}_{k}^{q}$ via solving $\mathcal{P}_2$ exploiting CVX, with fixed  $\left\{ \mathbf{w}_{k}^{q-1},\boldsymbol{\theta }^{q-1} \right\}$; \\
  Obtain $\boldsymbol{\theta}^{q}$ via solving $\mathcal{P}_5$ exploiting ADMM method, with fixed $\left\{ \mathbf{w}_{k}^{q},\boldsymbol{\theta }^{q-1} \right\}$; \\
  Update the fixed point $\left\{ \mathbf{w}_{k}^{q-1},\boldsymbol{\theta }^{q-1} \right\}$ exploiting the obtained point $\left\{ \mathbf{w}_{k}^{q},\boldsymbol{\theta }^{q} \right\}$; \\
  $q\leftarrow q + 1$;
  \UNTIL $R_s^{q} - R_s^{q-1} < \varepsilon _R$;
  \end{algorithmic}
\end{algorithm}

Ultimately, we transform  $\mathcal{P}_1$ into a solvable problem, with each of the two subproblems addressed using their respective methods. By integrating the aforementioned steps, we derive the complete BCD approach, which is summarized in Algorithm~\ref{BCD}. In this algorithm, $R_s^q$ represents the WSSR achieved in the $q$th iteration, while $\varepsilon _R$ is the convergence threshold. Furthermore, to ensure the convergence of Algorithm~\ref{BCD}, we present the following two theorems \cite{c32}.

\begin{theorem} \label{theorem1}
  The value of the objective function increases in each iteration of Algorithm~\ref{ADMM}, i.e., $R_{s}\left(\left\{\mathbf{w}_{k}^{q}\right\}_{k=1}^{K},{\boldsymbol{\theta}}^{q}\right) \le R_{s}\left(\left\{\mathbf{w}_{k}^{q+1}\right\}_{k=1}^{K},{\boldsymbol{\theta}}^{q+1}\right) $, which guarantees to converge to a locally optimal point.
\end{theorem}
\begin{proof}
  Kindly consult the Appendix \ref{appa}.
\end{proof}

\begin{theorem} \label{theorem2}
The solution  $ (\{\mathbf{w}_k^\ast\}_{k=1}^K,{\boldsymbol{\theta}}^\ast)$ will eventually converge to a Karush-Kuhn-Tucker (KKT) point.
\end{theorem}
\begin{proof}
  Kindly consult the Appendix \ref{appb}.
\end{proof}

\begin{remark}
{\textbf{Initialization and Sensitivity Discussion:}}
The proposed BCD-SCA-ADMM algorithm requires a feasible initial point. Specifically, the initial beamforming vectors should satisfy
\begin{equation*}
  \sum_{k=1}^K{\left\| \mathbf{w}_k \right\| ^2}\le P_T,
\end{equation*}
and the initial ARIS phase shift vector $\boldsymbol{\theta }^0$ should satisfy
\begin{equation*}
  |\theta _{n}^{0}|=1,\;\forall n\in \mathcal{N}.
\end{equation*}

A simple feasible initialization can be obtained by generating arbitrary nonzero beamforming vectors $\{\tilde{\mathbf{w}}_k^0\}_{k=1}^{K}$ and then normalizing them as
\begin{equation*}
  \mathbf{w}_k^0 =\sqrt{P_T}\frac{\tilde{\mathbf{w}}_k^0}{\sqrt{\displaystyle\sum_{j=1}^{K} \|\tilde{\mathbf{w}}_j^0\|^2}},~\forall k \in \mathcal{K}.
\end{equation*}

The ARIS phase shift can be initialized as $\theta_n^0=e^{j\phi_n^0},~\phi_n^0\in[0,2\pi),$ or simply $\theta_n^0=1, \forall n\in\mathcal N$. In the ADMM procedure, the auxiliary variable and dual variable can be initialized as $\boldsymbol{r}^0=\boldsymbol{\theta }^0$ and $\boldsymbol{p}^0=\mathbf{0}$, respectively.

Since the original WSSR maximization problem is non-convex, the proposed algorithm may converge to different local stationary points under different feasible initializations. Therefore, the final WSSR value can be affected by the initial point. Nevertheless, the convergence guarantee of the proposed algorithm does not depend on a particular initialization. As long as the initial point is feasible, the BCD iterations monotonically improve the objective value and converge to a KKT point, as discussed in Theorems \ref{theorem1} and \ref{theorem2}. Hence, the initialization mainly influences the convergence speed and the attained local stationary point, rather than the feasibility and convergence property of the proposed algorithm.
\end{remark}

\subsection{Complexity Analysis}

In this subsection, we provide the computational complexity analysis of the proposed BCD-based algorithm. In each outer BCD iteration, the main computational cost comes from two parts: the active beamforming optimization subproblem and the ARIS phase shift optimization subproblem.

First, for the active beamforming optimization subproblem, the SCA technique is used to construct a convex approximation around the local point. The dominant arithmetic operations are associated with the construction of the quadratic terms involving the effective cascaded channels and the beamforming vectors. Since there are $K$ users and $K$ beamforming vectors, the desired-signal and inter-user-interference terms involve $K^2$ user-interference pairs. For each pair, the main matrix operations involve the $M$-dimensional BS beamforming vector and the $N$-dimensional ARIS-related cascaded channel matrix, whose complexity is approximately $\mathcal{O}(M^2N)$. Therefore, the computational complexity of constructing the active beamforming subproblem is
$\mathcal{O}(M^2K^2N)$.

Second, for the ARIS phase shift optimization subproblem, the beamforming vectors are fixed and the phase shift vector $\boldsymbol{\theta } \in\mathbb C^N$ is optimized. The main cost is caused by constructing the quadratic matrix and vector in the phase shift subproblem. This step involves matrix products between the $N$-dimensional ARIS phase vector, the $M$-dimensional beamforming vectors, and the cascaded channel matrices. Considering all $K^2$ user-interference pairs, the dominant complexity of constructing the phase shift subproblem is $\mathcal{O}(MK^2N^2)$.

Then, the phase shift optimization subproblem is solved by the ADMM method. According to the ADMM updates, the $\boldsymbol{r}$-update mainly involves matrix-vector multiplications with an $N\times N$ matrix, while the $\boldsymbol{\theta}$-update is an element-wise projection onto the unit-modulus constraint. Hence, the per-iteration complexity of the ADMM update is approximately $\mathcal{O}(N^2)$. Let $T_{\rm ADMM}$ denote the number of ADMM iterations. The total complexity of the ADMM procedure is therefore
$\mathcal{O}(T_{\rm ADMM}N^2)$.

Combining the above three parts, the computational complexity of one BCD iteration is $\mathcal{O}\left(M^2K^2N+MK^2N^2+T_{\rm ADMM}N^2\right)$. Therefore, if $T_o$ denotes the number of outer BCD iterations, the overall computational complexity of Algorithm~\ref{BCD} is given by $\mathcal{O}\left(T_o\left(M^2K^2N+MK^2N^2+T_{ADMM}N^2 \right)\right) $.


\section{Extension of the Investigated System}  \label{add}
In this section, we provide additional details regarding the considered scenario. Specifically, it is difficult to enable continuous adjustment of the phase for the RC due to the constraints imposed by the actual hardware setup. Therefore, we discuss the discrete phase adjustment mode of the RC in Subsection \ref{drc}. Furthermore, given the limitations of practical conditions,  obtaining perfect CSI is challenging. Consequently, we discuss the case of imperfect CSI  in Subsection \ref{imcsi}.

\subsection{Discrete Phase Shift} \label{drc}
In the previous section, this study focuses on the scheme with continuous phase shift. In practice, the ARIS phase shift is typically discrete. Here, we redefine the continuous phase shift case as $\mathcal{F}_1$ and the discrete RC case as $\mathcal{F}_2$ \cite{r16}, i.e.,
\begin{subequations}
  \begin{align}
    &\mathcal{F}_{1}=\left\{\theta_{n}\left|\theta_{n}{=}e^{j\varphi_{n}},\varphi_{n}\in[0,2\pi)\right\},\right. \\
    &\mathcal{F}_{2}=\left\{\theta_{n}\left|\theta_{n}{=}e^{j\varphi_{n}},\varphi_{n}\in\left\{0,\frac{2\pi}{E},\ldots,\frac{2\pi(E-1)}{E}\right\}\right\},\right.
  \end{align}
\end{subequations}
where $E = 2^b$ denotes that $\mathcal{F}_{2}$ has $E$ discrete phase shift values, and $b$ is the number of quantization bits.

In fact, to solve the discrete RC problem, the obtained $\theta_{n} \in \mathcal{F}_{1}$ can be directly applied to the discrete phase shift case, where $\theta_{n} \in \mathcal{F}_{2}$. Specifically, we denote the solution of $\theta_{n} \in \mathcal{F}_{1}$ and $\theta_{n} \in \mathcal{F}_{2}$ as $\theta_{n}^{\prime}$ and $\theta_{n}^{''} $, respectively, and relax $\theta_{n}^{\prime}$ to $\theta_{n}^{''} $, and then derive the suboptimal solution $\boldsymbol{\theta}^\ast$ during each iteration to optimize the continuous $\boldsymbol{\theta}$, thus obtaining the suboptimal $\theta_{n}^{\prime,\ast}$. Next, we map $\theta_{n}^{\prime,\ast}$ onto $\mathcal{F}_{2}$ to derive $\theta_{n}^{''} $, i.e., $\theta_{n}^{''} = e^{j\varphi_{l^{\ast}}} $, where $l^{\star}=\underset{1\leq l \leq E}{\operatorname*{\mathrm{arg}\operatorname*{min}}}\left|\theta_{n}^{\prime,\ast}-e^{j\phi_{l}}\right|$. Thus, we can derive the suboptimal solution for  phase shift optimization in each iteration of the subproblem.

For the discrete RC case, we cannot apply Theorems \ref{theorem1} and \ref{theorem2} directly. However, according to \cite{r16}, as long as the discrete phase shift $\boldsymbol{\theta}^q$ satisfies $f\left({\boldsymbol{\theta}}^q \right) \ge f\left({\boldsymbol{\theta}}^{q-1} \right)$, where $f\left({\boldsymbol{\theta}}^q \right)$ is the objective value of $\mathcal{P}_6$ in the $q$th iteration, the proposed algorithm will converge in the case of discrete RC. The simulation in Section \ref{nr} demonstrates that the discrete RC case also exhibits strong convergence.

\subsection{Imperfect CSI} \label{imcsi}
In the proposed ARIS-assisted HST communication systems, channel estimation is particularly challenging due to the high-dimensional nature of the ARIS channels and the passive characteristics of the Eve. In this context, we analyze the robustness of the proposed joint beamforming approach to CSI error. We assume that all the channels involved are imperfect and can be written as \cite{r17}
\begin{equation}
  h = \tilde{h} + \Delta h.
\end{equation}
where $h$ denotes the actual channel, $\tilde{h}$ represents the estimated channel, and $\Delta h$ signifies the  estimation error with Gaussian distribution and zero mean, i.e., $\Delta h\sim\mathcal{CN}\left(0,\sigma_d^2\right)$. The variance $\sigma_d^2$ satisfies $\sigma_d^2=\delta_d\left|\hat{h}\right|^2$, where $\delta_d$ represents the ratio of the error power $\sigma_d^2$ to the channel gain $\left|\hat{h}\right|^2$, characterizing the CSI error level of the channel estimation. Then, the imperfect CSI is incorporated into the aforementioned BCD algorithm.

\section{Numerical Results} \label{nr}
In this section, we present a numerical evaluation of the performance of the proposed secure ARIS-assisted downlink MU-MISO system for HST communications. The results of the numerical simulations demonstrate the efficacy and benefits of the proposed ARIS-assisted WSSR maximization technique for HST communication systems.

\subsection{Simulation Setup}
The following simulation parameters are assumed: $M = 6$, $N = 100$, $K = 4$, $f_c = 28$ GHz, $P_T = 30$ dBm, $v=360~{\rm{km/h}}$ , $\sigma_{k}^2 = \sigma_{e}^2 = -174~ {\rm{dBm/Hz}} +10\log _{10}B+10 ~\rm{dB}$, where $B = 200$ MHz is the system bandwidth, and the user's weight is assigned as $\varrho_k = 1/K$. The height of ARIS $H_{\rm{R}} = 100$ m, the height of BS $H_{\rm{B}} = 10$ m. This paper considers the railway track as the x-axis of a three-dimensional coordinate system, and the users on the carriage move along the x-axis. Therefore, only the horizontal positions of the users and the eavesdroppers are considered, and their height is 0, i.e., ${H_{\mathrm{U}}=H_{\mathrm{E}}=0}$, the Rician $K$-factor $\kappa_{\rm{R,k}} = \kappa_{\rm{RE}} = \kappa_{\rm{BR}} = 10$ dB, the path loss exponent $\alpha_{\rm{R,k}} = \alpha_{\rm{RE}} = 2.5$, and $\alpha_{\rm{BR}} = 2$.

For comparison, we consider the following baseline cases:
\begin{itemize}
  \item {\bf Ideal Phase Shift:} The phase shift of the ARIS is synchronized with the phase shift of the cascaded channel.
  \item {\bf Imperfect CSI:} All the channels are imperfect, and we incorporate this into the proposed algorithm.
  \item {\bf Discrete Phase Shift:} Discrete phase shift while executing  the procedures outlined in Section \ref{drc}.
  \item {\bf Random Beamforming:} The initial beamforming vector is random, and only the phase shift is optimized without optimizing the beamforming vector.
  \item {\bf Random Phase Shift:} The initial phase shift is random, with optimization applied solely to the beamforming vector, leaving the phase shift unoptimized.
\end{itemize}

\subsection{Convergence of proposed Algorithm}
Here, this study confirms the convergence characteristics of the proposed scheme designed for the joint optimization of active beamforming and discrete ARIS phase shift. From the results of Fi{}g.~{\ref{fig:2}}, the proposed scheme demonstrates effectiveness, illustrating the convergence rates across all evaluated schemes. Notably, the proposed algorithm significantly outperforms both the random beamforming and random phase shift schemes. For optimal beamforming and ARIS phase shift, the WSSR optimization problem involves two key objectives: maximizing the rate for legitimate users and minimizing the rate for Eve. However, the random beamforming scheme optimizes only the ARIS phase shift, while the random phase shift scheme optimizes only beamforming, resulting in suboptimal performance compared to the proposed scheme. This underscores the advantage of jointly optimizing both beamforming and ARIS phase shift for improved performance. 
\begin{figure}[!t]
  \centering
  {\includegraphics[scale=0.5]{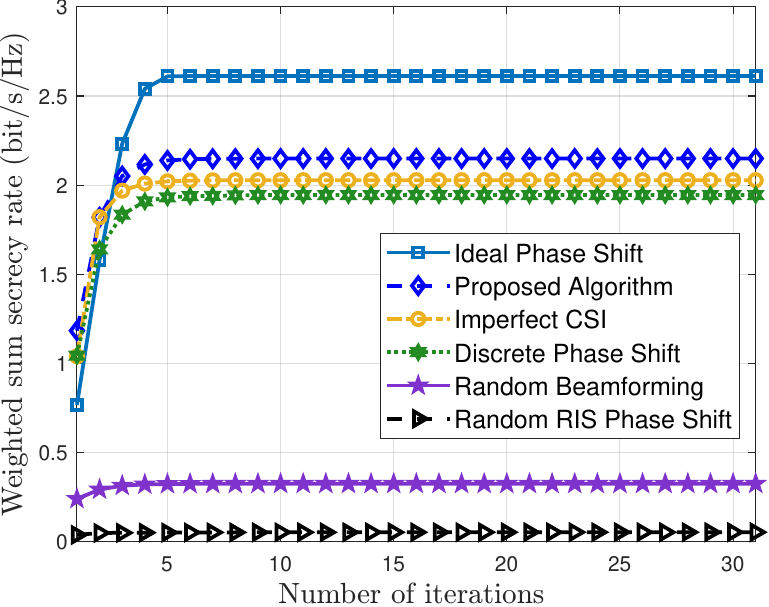}}
  \caption{ \label{fig:2}Weighted sum secrecy rate versus number of iterations.}
\end{figure}

\begin{figure}[!t]
  \centering
  {\includegraphics[scale=0.5]{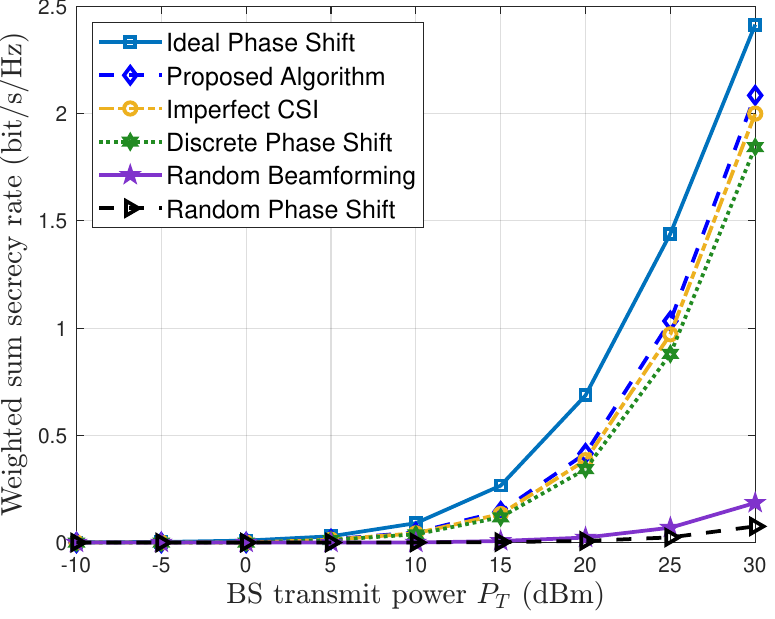}}
  \caption{ \label{fig:3}Weighted sum secrecy rate versus transmit power $P_{T}$.}
\end{figure}
\subsection{Impact of Transmit Power}
Fi{}g.~{\ref{fig:3}} shows the WSSR versus the transmit power $P_{T}$ for different schemes. The results demonstrate that the WSSR increases rapidly as the BS transmit power rises across all six cases. 
Additionally, we observe that the performance gap between the ideal phase shift, the proposed algorithm, the imperfect CSI, the discrete phase, the random beamforming, and the random phase shift schemes widens as $P_{T}$ increases. A key finding is that the performance gain achieved through optimized beamforming and ARIS phase shift is substantial only when $P_{T}$ is high (e.g., from $10$ dBm to $30$ dBm), whereas it becomes negligible at lower values of $P_{T}$ (e.g., from $-10$ dBm to $5$ dBm). This occurs because, at low $P_{T}$, the signals reflected by the ARIS are weak, contributing minimally to performance improvement. Conversely, when $P_{T}$ is sufficiently high, the BS allocates more power to the BS-ARIS-user link. In this scenario, the rate for legitimate users improves significantly, while the rate for Eve is more effectively suppressed, resulting in a notable increase in WSSR. 
However, when $P_{T}$ is sufficiently large, substantial performance improvements can be achieved with optimized beamforming and ARIS phase shift. In such cases, the BS tends to direct  most of its power toward the BS-ARIS-user link, thereby weakening the BS-ARIS-Eve link. 

\begin{figure}[!t]
  \centering
  {\includegraphics[scale=0.5]{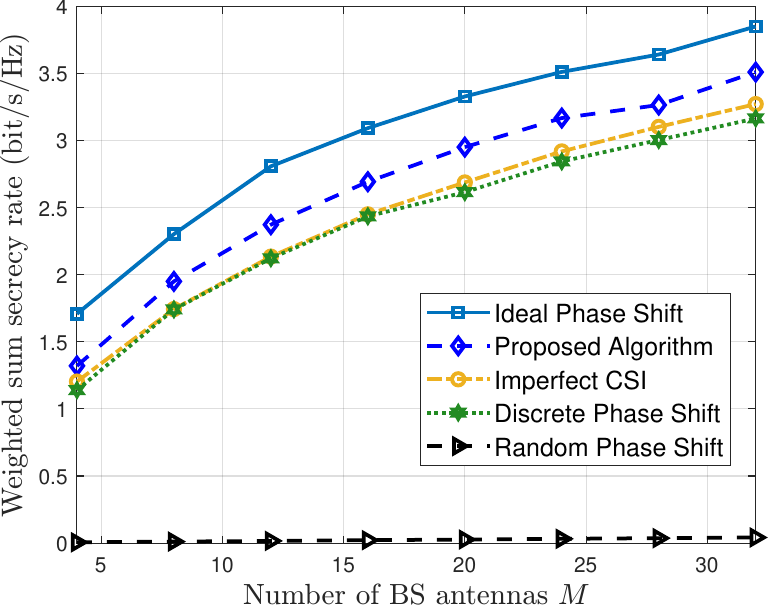}}
  \caption{ \label{fig:4}Weighted sum secrecy rate versus number of antennas $M$.}
\end{figure}
\subsection{Impact of Number of Antennas}
In Fi{}g.~{\ref{fig:4}}, we evaluate the effect of the number of BS transmit antennas $M$. As anticipated, increasing $M$ results in higher WSSR values for all algorithms. This phenomenon can be attributed to the fact that an increase in $M$ not only provides greater diversity gain for legitimate users but also introduces more degrees of freedom for Eve, leading to a simultaneous increase in WSSR. Under the same simulation conditions, when $M = 8$, $M = 16$, $M = 24$ and $M = 32$, the WSSR of the proposed scheme is approximately 181.14-fold, 125.75-fold, 99.84-fold, and 84.06-fold higher than that of the random phase shifting scheme, respectively. This indicates that as $M$ increases, the performance gap between the proposed algorithm and the random phase shift decreases. 
\begin{figure}[!t]
  \centering
  {\includegraphics[scale=0.5]{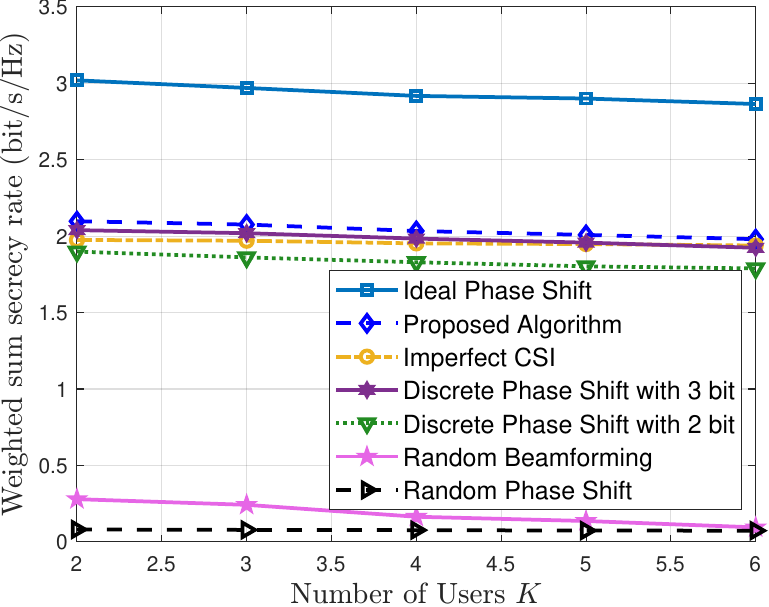}}
  \caption{ \label{fig:5}Weighted sum secrecy rate versus number of users $K$}.
\end{figure}
\subsection{Impact of Number of Users}
In Fi{}g.~{\ref{fig:5}}, we evaluate the effect of the number of users $K$ on WSSR. As anticipated, the WSSR decreases with $K$ for six different cases. The proposed algorithm consistently surpasses both the random beamforming and the random phase shift schemes. 
This decline in WSSR is mainly due to two reasons. First, since we set the weight for each user as ${\varrho _k=1/K}$, the contribution of each individual secrecy rate to the WSSR decreases as $K$ increases. Second, with more users being served simultaneously, the inter-user interference becomes more severe, and the limited transmit power and ARIS reflecting resources need to be shared among more users. As a result, the WSSR performance decreases. Furthermore, the results indicate that WSSR increases as the number of ARIS quantization bits $b$ increases. This improvement stems from the enhanced phase resolution provided by a larger $b$, which increases the channel gain for the users and subsequently improves the WSSR. However, this improvement entails a trade-off, resulting in higher overhead and more complex hardware requirements. Additionally, we observe that $3$ bit discrete phase shift offers performance closer to that of the proposed scheme compared to $2$ bit discrete phase shift. This implies that as $b$ increases, the performance gap between the proposed algorithm and the discrete phase shift narrows.

\begin{figure}[!t]
	\centering
	{\includegraphics[scale=0.5]{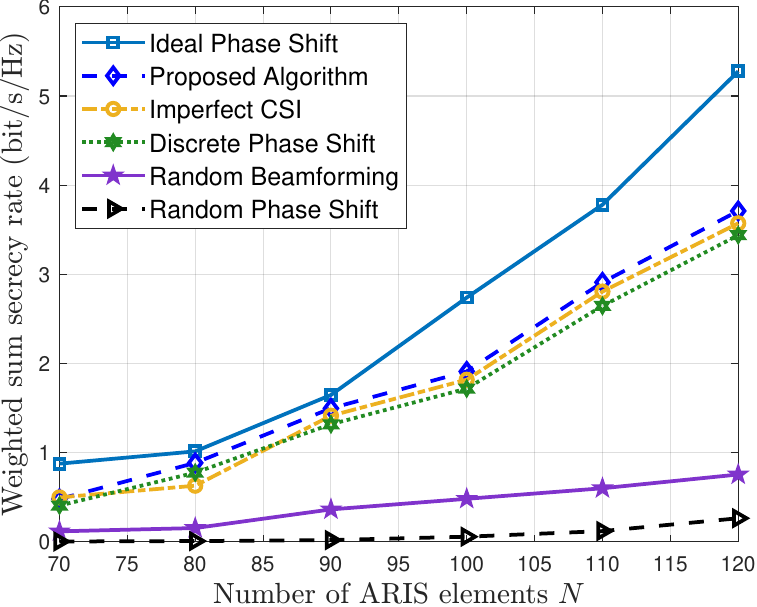}}
	\caption{ \label{fig:6}Weighted sum secrecy rate versus number of ARIS elements $N$.}
\end{figure}
\subsection{Impact of Number of ARIS Elements}
Fi{}g.~{\ref{fig:6}} illustrates the WSSR versus the number of ARIS elements $N$ for six different cases. The results indicate that the cases where both beamforming and ARIS phase shift are optimized outperform the scenarios where only one of these components is optimized. Specifically, the WSSR increases as $N$ increases. This enhancement can be credited to the increase in the number of ARIS elements, which reflects more signals toward the legitimate user, resulting in a stronger received signal. Furthermore, it is evident that the performance gap among the four cases, namely, ideal phase shift, proposed scheme, imperfect CSI, and discrete phase shift, and the two schemes, namely, random beamforming and random phase shift, will widen as $N$ increases. This phenomenon arises because, with more ARIS elements in the system, the ARIS can enhance signal quality by reflecting additional signal paths and energy for legitimate users while more effectively diminishing the signal quality for Eve.

\begin{figure}[!t]
  \centering
  {\includegraphics[scale=0.5]{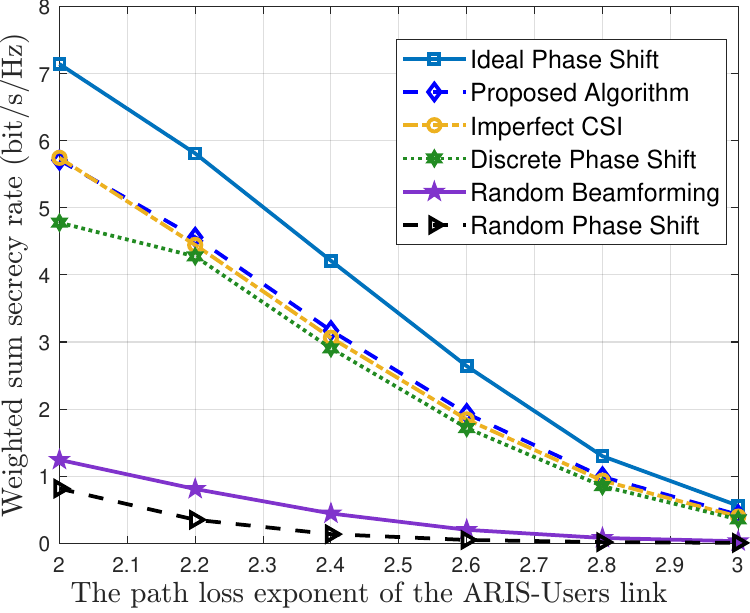}}
  \caption{ \label{fig:7}Weighted sum secrecy rate versus the path loss exponent of the ARIS-user link.}
\end{figure}
\subsection{Impact of Path Loss Exponent of the ARIS-user link}
In Fi{}g.~{\ref{fig:7}}, we evaluate the path loss exponent of the ARIS-user link on WSSR for the six schemes. The results demonstrate that the WSSR decreases as the path loss exponent increases across all six schemes. This phenomenon occurs because higher path loss in the ARIS-user link leads to more severe fading, resulting in a weaker signal reflected by the ARIS.  Consequently, the effectiveness of optimized beamforming and ARIS phase shift is reduced. Furthermore, the performance gap between the four schemes, namely, the ideal phase shift, the proposed algorithm, the imperfect CSI and the discrete phase shift, and the two schemes, namely, the random beamforming and the random phase shift, narrows as the exponent increases, supporting the aforementioned explanation. If the exponent continues to increase, the WSSR for all schemes may approach zero. This indicates that when the path loss exponent of the ARIS-user link is sufficiently high, the WSSR of schemes optimizing both beamforming and ARIS phase shift converges to that the schemes optimizing only one of these factors.

\begin{figure}[!t]
  \centering
  {\includegraphics[scale=0.5]{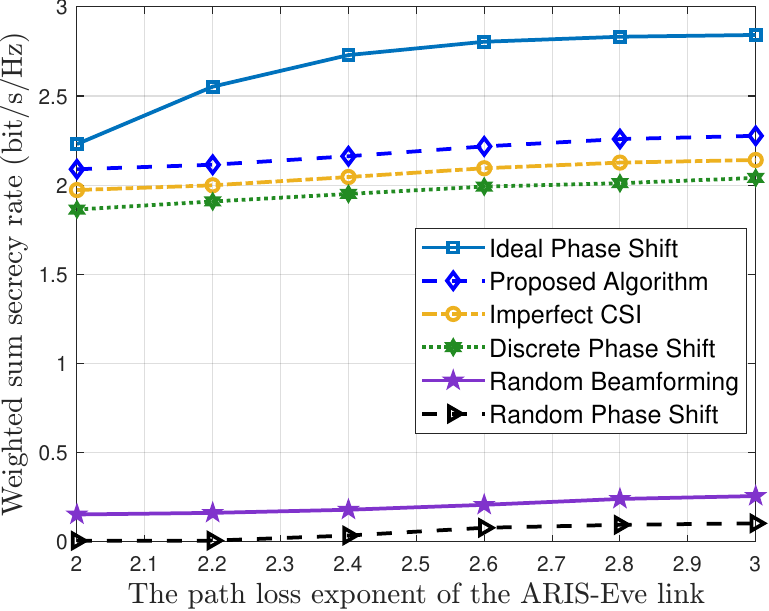}}
  \caption{ \label{fig:8}Weighted sum secrecy rate versus the path loss exponent of the ARIS-Eve link.}
\end{figure}
\subsection{Impact of Path Loss Exponent of the ARIS-Eve Link}
Fi{}g.~{\ref{fig:8}} illustrates the influence of path loss exponent of the ARIS-Eve link on the WSSR across the six schemes. The results demonstrate that the WSSR rises as the path loss exponent of the ARIS-Eve link increases. This phenomenon occurs because stronger large-scale fading attenuates the signal reflected toward Eve, thereby degrading Eve's reception quality. These findings underscore the importance of strategically deploying the ARIS to minimize obstacles in legitimate communication links or to introduce additional obstacles in Eve links, thereby enhancing WSSR performance. Specifically, the ARIS should be deployed to maximize the rate for legitimate users while simultaneously minimizing Eve's rate, ultimately improving the WSSR.

\begin{figure}[!t]
	\centering
	{\includegraphics[scale=0.5]{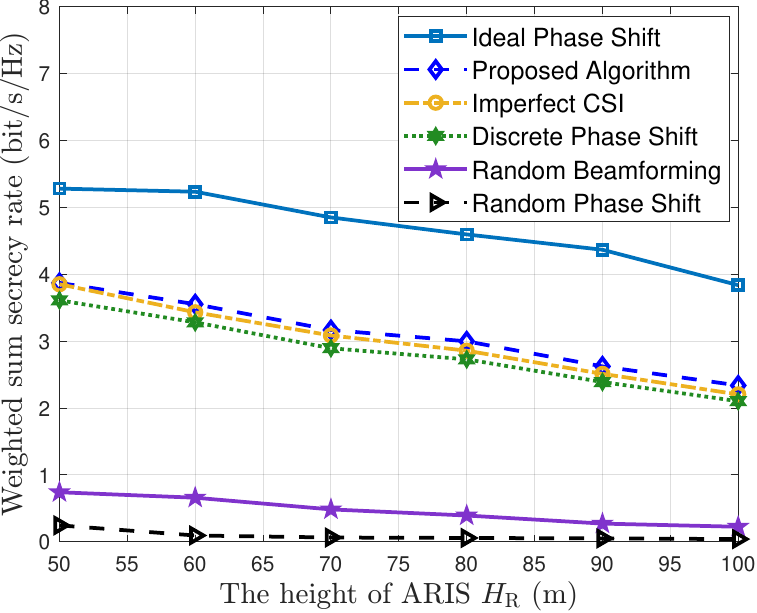}}
	\caption{ \label{fig:9}Weighted sum secrecy rate versus the height of ARIS.}
\end{figure}
\subsection{Impact of Height of ARIS}
Fi{}g.~{\ref{fig:9}} depicts the impact of the ARIS height $H_{\mathrm{R}}$ on the WSSR of the six schemes. Under the adopted distance-dependent Rician channel model with fixed Rician $K$-factor, increasing $H_{\mathrm{R}}$ enlarges the slant distances of the BS-ARIS and ARIS-user links, thereby reducing the corresponding large-scale channel gains and weakening the cascaded legitimate links. Since the adopted model does not include an altitude-dependent LoS-probability gain, the distance-induced attenuation dominates in this setting, and the WSSR decreases with $H_{\mathrm{R}}$. We emphasize that this trend is specific to the adopted channel model and simulation geometry, rather than a universal conclusion for all UAV communication channels. In particular, under elevation-angle-dependent probabilistic LoS/NLoS models, the performance may exhibit non-monotonic behavior with respect to the ARIS height. In our considered scenario, however, the simplified Rician model is used to highlight the effect of joint active beamforming and ARIS phase shift optimization in a severely blocked HST environment. 

\subsection{Impact of K-factor}
Fi{}g.~{\ref{fig:10}} shows the impact of the Rician K-factor on WSSR, which varies from $5$ dB to $12$ dB. The results show that the WSSR increases for all six schemes as $\kappa$ increases. This occurs because a larger Rician K-factor results in a stronger LoS component, thereby improving the overall channel quality between the BS and the legitimate users. The enhanced channel quality leads to a stronger received signal for legitimate users, which consequently increases the WSSR. When $\kappa \to \infty$, the NLoS component becomes negligible, and the channel model converges to a deterministic LoS channel. Conversely, as $\kappa \to 0$, the model degenerates into a Rayleigh fading channel, where the LoS component is absent.
\begin{figure}[!t]
	\centering
	{\includegraphics[scale=0.5]{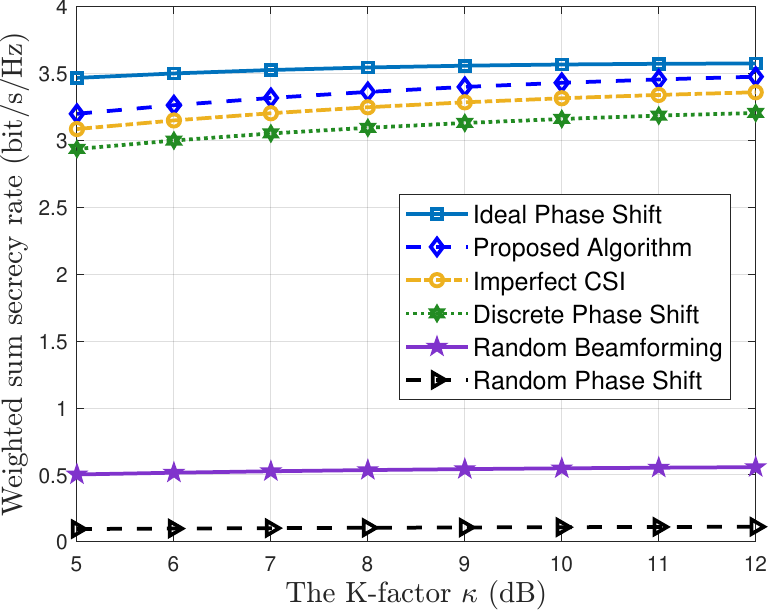}}
	\caption{ \label{fig:10}Weighted sum secrecy rate versus K-factor $\kappa$.}
\end{figure}

\begin{figure}[!t]
	\centering
	{\includegraphics[scale=0.5]{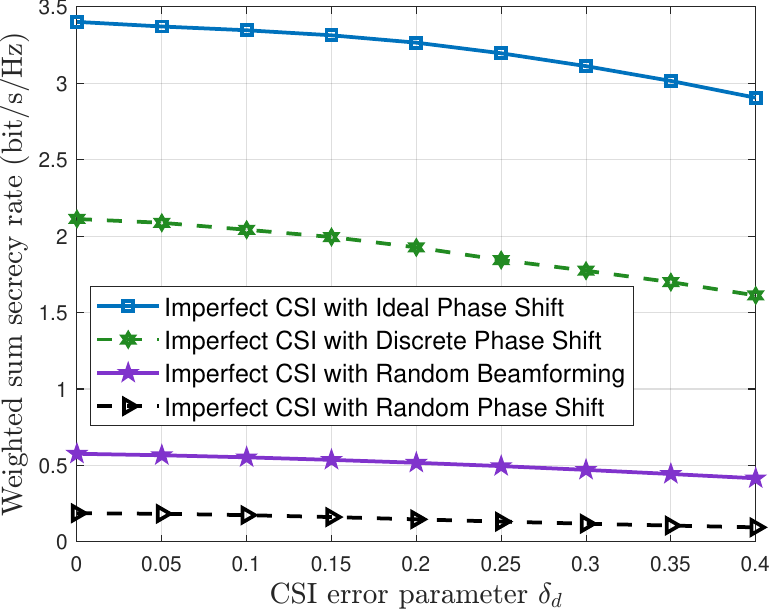}}
	\caption{ \label{fig:11}Weighted sum secrecy rate versus the CSI error level $\delta_d$.}
\end{figure}
\subsection{Impact of CSI Error Level}
Fi{}g.~{\ref{fig:11}} illustrates the impact of the CSI error level $\delta_d$ on the WSSR  of the six schemes, with $\delta_d$ varying from 0 to 0.4. The simulation results show that as the CSI error level increases, the WSSR for all algorithms decreases. Notably, for the ideal phase shift scheme, the WSSR experiences a reduction of $2.58\%$ when the error power is $15\%$ of the channel gain (i.e., $\delta_d = 0.15$), compared to the case with perfect CSI ($\delta_d = 0$). When the CSI error level increases to $\delta_d = 0.35$, the WSSR loss grows to $11.36\%$. This demonstrates that the proposed scheme
exhibits considerable robustness to CSI errors.

\subsection{Quantification of Link Connectivity}
To further clarify the practical service scope of the proposed ARIS-assisted HST communication system, we quantify the link connectivity along the target railway segment. Different from conventional static communication scenarios, the ARIS-user distance changes as the HST moves along the railway track. Therefore, the secure service capability of the ARIS-assisted system is position-dependent.

Let $x$ denote the longitudinal position of the HST relative to the projection point of the ARIS on the railway track. For each position $x$ the ARIS-user distance is updated according to the system geometry as
\begin{equation*}
  d_{\mathrm{R},k}(x)=\sqrt{x^2+d_\perp^2+(H_{\mathrm{R}}-H_{\mathrm{U}})^2},
\end{equation*}
where $d_\perp $ is the horizontal offset between the ARIS projection point and the railway track. Based on the updated distance, the corresponding ARIS-user channel is generated according to the adopted distance-dependent path loss and Rician fading models. Then, the proposed joint active beamforming and ARIS phase shift optimization algorithm is applied to obtain the optimized WSSR at position $x$, which is given by
\begin{equation*}
  R_s^{\star}(x)=\sum_{k=1}^{K}\varrho_k R_{s,k}^{\star}(x).
\end{equation*}

Since the objective of this paper is WSSR maximization, we use the optimized WSSR as the metric for quantifying secure service availability. Specifically, the ARIS-assisted secure service is regarded as available at position $x$ when
\begin{equation*}
  R_s^{\star}(x)\ge R_{\rm WSSR}^{\rm th},
\end{equation*}
where $R_{\rm WSSR}^{\rm th}$ denotes the required WSSR threshold. Accordingly, the secure connectivity indicator is defined as
\[
\mathcal{C}_{\mathrm{sec}}(x)=
\begin{cases}
1, & R_s^{\star}(x)\ge R_{\mathrm{WSSR}}^{\mathrm{th}},\\
0, & \text{otherwise}.
\end{cases}
\]
Here, $\mathcal C_{\rm sec}(x)=1$ means that the ARIS-assisted system can provide secure service at the HST position $x$, while $\mathcal C_{\rm sec}(x)=0$ means that the WSSR requirement cannot be satisfied at that position.

Fi{}g.~{\ref{fig:12}} shows the optimized WSSR versus the HST longitudinal position $x$. The horizontal dashed line represents the required threshold $R_{\rm WSSR}^{\rm th}$. The two boundary positions $x_1$ and $x_2$ are obtained from
\begin{equation*}
  R_s^{\star}(x_1)=R_s^{\star}(x_2)=R_{\rm WSSR}^{\rm th}.
\end{equation*}
Therefore, the maintainable secure connectivity length can be calculated as
\begin{equation*}
  D_{\rm conn}=x_2-x_1.
\end{equation*}
The corresponding service duration is
\begin{equation*}
  T_{\rm conn}=\frac{D_{\rm conn}}{v}.
\end{equation*}

\begin{figure}[!t]
	\centering
	{\includegraphics[scale=0.5]{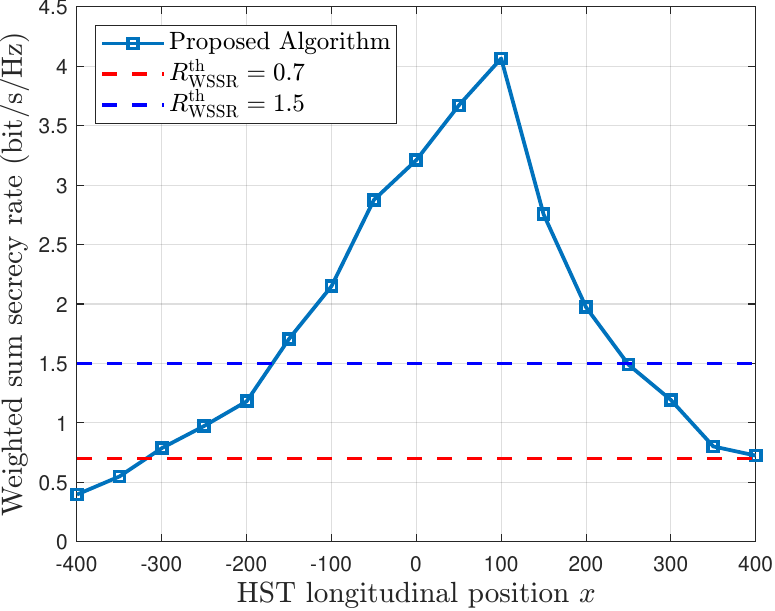}}
	\caption{ \label{fig:12}Weighted sum secrecy rate versus the HST longitudinal position $x$.}
\end{figure}

As shown in Fi{}g.~{\ref{fig:12}}, the optimized WSSR is relatively high when the HST is located within the effective reflected-coverage region of the ARIS, and it gradually decreases as the HST moves away from this region. Under the assumed system parameters, when $R_{\rm WSSR}^{\rm th} = 1.5$ bit/s/Hz, the proposed ARIS-assisted scheme can maintain secure service over approximately $430$ m, corresponding to a service duration of approximately $4.3$ s at $v=360$ km/h. When the threshold is relaxed to $R_{\rm WSSR}^{\rm th} = 0.7$ bit/s/Hz, the maintainable secure connectivity length increases to approximately $720$ m.


\section{Conclusions} \label{con}

This paper investigated secure coverage enhancement for an ARIS-assisted HST communication system. Instead of claiming a new generic optimization model for RIS-assisted secure systems, this work focused on the scenario-specific WSSR maximization problem under ARIS-assisted HST communications. A BCD-based joint optimization algorithm was proposed, where the active beamforming and ARIS phase shift were updated using SCA and ADMM, respectively. The proposed design was further extended to imperfect CSI and discrete phase-shift cases. Numerical results verified that joint active beamforming and ARIS phase shift optimization can effectively improve the secrecy performance compared with benchmark schemes. 
In future work, we will further investigate the statistical robustness and practical implementation of ARIS-assisted secure HST transmission, including multi-run performance evaluation, confidence-interval characterization, robust beamforming under imperfect Doppler compensation, channel aging, outdated CSI, and more practical UAV deployment constraints such as trajectory control, energy limitation, synchronization, duty-cycle-aware service scheduling, and timetable-based ARIS activation.

\appendices

\section{The proof of Theorem 1} \label{appa}
Initially, when updating $\{\mathbf{w}_{k}\}_{k=1}^{K}$ in the $(q+1)$th iteration, an upper bound $R_s^{ub1}$ for the WSSR $R_s$ around given point $\left(\left\{\mathbf{w}_{k}^{q}\right\}_{k=1}^{K},{\boldsymbol{\theta}}^{q}\right)$ is expressed as 
  \begin{align} 
  & R_s^{ub1}\left(\left\{\mathbf{w}_{k}\right\}_{k=1}^{K},{\boldsymbol{\theta}}^{q}\right) \nonumber \\
  & = \sum_{k=1}^K{\varrho _k \left\{  \frac{c_12\Re \{\mathbf{w}_{k}^{\mathrm{H}}\bar{\mathbf{H}}_{k}^{\mathrm{H}}\boldsymbol{\theta }^q\left( \boldsymbol{\theta }^q \right) ^{\mathrm{H}}\bar{\mathbf{H}}_k\mathbf{w}_{k}^{q}\}}{\ln 2} \right.}   \nonumber \\
   &-\frac{c_2 \left( \sum\nolimits_{j=1,j\ne k}^K  {\left| \left( \boldsymbol{\theta }^q \right) ^{\mathrm{H}}\bar{\mathbf{H}}_k\mathbf{w}_j \right|^2}+1 +\left| \left( \boldsymbol{\theta }^q \right) ^{\mathrm{H}}\bar{\mathbf{H}}_k\mathbf{w}_k \right|^2 \right)}{\ln 2}  \nonumber \\  
   &-\frac{z_{e,k}}{\left( 1+\tilde{z}_{e,k} \right) \ln 2}+\frac{2\Re \left\{ \mathbf{W}^{\mathrm{H}}\Lambda _{e,k}\Lambda _{e,k}^{\mathrm{H}}\mathbf{W}^q \right\}}{\ln 2}  \nonumber  \\ 
   &\left. -\frac{\left\| \Lambda _{e,k}^{\mathrm{H}}\mathbf{W}^q \right\| ^2\left\| \Lambda _{e,k}^{\mathrm{H}}\mathbf{W} \right\| ^2}{\left( 1+\left\| \Lambda _{e,k}^{\mathrm{H}}\mathbf{W}^q \right\| ^2 \right) \ln 2} + L\right\}, 
\end{align}
where $L$ is the constant term, which can be given as
\begin{align}
  L=&\log_2\left( 1+c_1\left| \left( \boldsymbol{\theta }^q \right) ^{\mathrm{H}}\bar{\mathbf{H}}_k\mathbf{w}_{k}^{q} \right|^2 \right) -\frac{c_1\left| \left( \boldsymbol{\theta }^q \right) ^{\mathrm{H}}\bar{\mathbf{H}}_k\mathbf{w}_{k}^{q} \right|^2}{\ln 2} \nonumber \\
 &-\log_2\left( 1+\tilde{z}_{e,k} \right) +\frac{\tilde{z}_{e,k} }{\left( 1+\tilde{z}_{e,k} \right) \ln 2}.
\end{align}

Similarly, during the update of $\boldsymbol{\theta }$ at the $(q+1)$th iteration, the following upper bound $R_s^{ub2}$ of $R_s$ around given point $\left(\left\{\mathbf{w}_{k}^{q+1}\right\}_{k=1}^{K},{\boldsymbol{\theta}}^{q}\right)$ can be expressed as
\begin{align} 
  &R_s^{ub2}\left(\left\{\mathbf{w}_{k}^{q+1}\right\}_{k=1}^{K},{\boldsymbol{\theta}}\right) \nonumber \\
  & = \sum_{k=1}^K{\varrho _k\left\{ \frac{c_12\Re \{\boldsymbol{\theta }^{\mathrm{H}}\bar{\mathbf{H}}_k\mathbf{w}_{k}^{q+1}\left( \mathbf{w}_{k}^{q+1} \right) ^{\mathrm{H}}\bar{\mathbf{H}}_{k}^{\mathrm{H}}\boldsymbol{\theta }^q\}}{\ln 2} \right.} \nonumber \\
  & -\frac{c_2\left( \sum\nolimits_{j=1,j\ne k}^K {\left| \boldsymbol{\theta }^{\mathrm{H}}\bar{\mathbf{H}}_k\mathbf{w}_{j}^{q+1} \right|^2}+1 +\left| \boldsymbol{\theta }^{\mathrm{H}}\bar{\mathbf{H}}_k\mathbf{w}_{k}^{q+1} \right|^2 \right)}{\ln 2} \nonumber \\
  & -\frac{\sum\nolimits_{j=1}^K{\left| \boldsymbol{\theta }^{\mathrm{H}}\bar{\mathbf{H}}_{\mathrm{E}}\mathbf{w}_{j}^{q+1} \right|^2}}{\left( 1+\tilde{z}_{e,k} \right) \ln 2}+\frac{2\Re \left\{ \boldsymbol{\theta }^{\mathrm{H}}\Upsilon _{e,k}\Upsilon _{e,k}^{\mathrm{H}}\boldsymbol{\theta }^q \right\}}{\ln 2}  \nonumber    \\
  & \left. -\frac{\left\| \Upsilon _{e,k}^{\mathrm{H}}\boldsymbol{\theta }^q \right\| ^2\left\| \Upsilon _{e,k}^{\mathrm{H}}\boldsymbol{\theta } \right\| ^2}{\left( 1+\left\| \Upsilon _{e,k}^{\mathrm{H}}\boldsymbol{\theta }^q \right\| ^2 \right) \ln 2} + L \right\}.
\end{align}

With $R_s^{ub1}\left(\left\{\mathbf{w}_{k}\right\}_{k=1}^{K},{\boldsymbol{\theta}}^{q}\right)$ and $R_s^{ub2}\left(\left\{\mathbf{w}_{k}^{q+1}\right\}_{k=1}^{K},{\boldsymbol{\theta}}\right)$, we obtain
\begin{align}
  &R_s\left( \left\{ \mathbf{w}_{k}^{q} \right\}_{k=1}^{K},{\boldsymbol{\theta}}^q \right) \underset{(\mathrm{a)}}{\le}\max_{\left\{ \mathbf{w}_k \right\} _{k=1}^{K}} R_{s}^{ub1}\left( \left\{ \mathbf{w}_k \right\} _{k=1}^{K},{\boldsymbol{\theta}}^q \right) \nonumber \\
 & \underset{(\mathrm{b)}}{\le}\max_{\hat{\boldsymbol{\theta}}} R_{s}^{ub2}\left( \left\{ \mathbf{w}_{k}^{q+1} \right\} _{k=1}^{K},{\boldsymbol{\theta}} \right)  \le R_s\left( \left\{ \mathbf{w}_{k}^{q+1} \right\} _{k=1}^{K},{\boldsymbol{\theta}}^{q+1} \right),
\end{align}
where $\left(a\right)$ and $\left(b\right)$ hold since $\left\{\mathbf{w}_{k}^{q+1}\right\}_{k=1}^{K}$ and $\boldsymbol{\theta}^{q+1}$ are the optimal solutions of the convex problems $\mathcal{P}_2$ and $\mathcal{P}_5$, respectively.

Furthermore, due to \eqref{yyaa2} and \eqref{yyaa3}, we know that $\left\{\mathbf{w}_{k}^{q}\right\}_{k=1}^{K}$ and $\boldsymbol{\theta}^{q}$ are both bounded. According to  Cauchy's theorem \cite{r14}, the sequence $\left(\left\{\mathbf{w}_{k}^{q}\right\}_{k=1}^{K}, \boldsymbol{\theta}^{q}\right)$ converges to $\left(\left\{\mathbf{w}_{k}^{\ast}\right\}_{k=1}^{K}, \boldsymbol{\theta}^{\ast}\right)$ as $q\rightarrow \infty $, i.e.,
\begin{align}
  & 0=\lim_{q\to\infty} \left\{ R_{s} \left(\left\{ \mathbf{w}_{k}^{q}\right\}_{k=1}^{K},{ \boldsymbol{\theta}}^{q}\right) - R_{s}\left( \left\{\mathbf{w}_{k}^{\ast} \right\}_{k=1}^{K},{\boldsymbol{\theta}}^{\ast}\right) \right\} \nonumber    \\ 
   &\leq\lim_{q\to\infty}\left\{R_{s}\left(\left\{\mathbf{w}_{k}^{q+ 1}\right\}_{k=1}^{K},{\boldsymbol{\theta}}^{q+1}\right) - R_{s}\left(\left\{\mathbf{w}_{k}^{\ast}\right\}_{k= 1}^{K},{\boldsymbol{\theta}}^{\ast}\right)\right\}.
\end{align}
Therefore, we have demonstrated that $ R_{s}\left(\left\{\mathbf{w}_{k}^{q}\right\}_{k= 1}^{K},{\boldsymbol{\theta}}^{q}\right) \leq  R_{s}\left(\left\{\mathbf{w}_{k}^{q+1}\right\}_{k= 1}^{K},{\boldsymbol{\theta}}^{q+1}\right)$, which can guarantee to converge to a locally optimal point.

\section{The proof of Theorem 2} \label{appb}
As stated earlier, the sequence $\left(\left\{\mathbf{w}_{k}^{q}\right\}_{k=1}^{K}, \boldsymbol{\theta}^{q}\right)$ will converge to $\left(\left\{\mathbf{w}_{k}^{\ast}\right\}_{k=1}^{K}, \boldsymbol{\theta}^{\ast}\right)$ as $q\rightarrow \infty $. Next, the Lagrangian function of $\mathcal{P}_2$ is given as
\begin{align}
  \mathcal{L}_{1}\left(\left\{\mathbf{w}_{k}\right\}_{k=1}^{K},{\boldsymbol{\theta }}^{q},\zeta\right)&=R_{s}^{ub1}\left(\left\{\mathbf{w}_{k}\right\}_{k=1}^{K},{\boldsymbol{\theta}}^{q}\right) \nonumber \\
  &\quad+\zeta\left(\sum_{k=1}^K\left\|\mathbf{w}_k\right\|^2-P_T\right),
\end{align}
where $\zeta \geq 0 $ is the dual variable  in \eqref{yyaa2}.

Then, when $q\rightarrow \infty $, the related KKT conditions are given as follows:
\begin{align}
  \left\{ 
    \begin{array}{l}
    \nabla _{\mathbf{w}_k}R_{s}^{ub1}\left\{ \mathbf{w}_{k}^{\ast}\}_{k=1}^{K},{\boldsymbol{\theta}}^q \right) +2\zeta^{\ast}\mathbf{w}_{k}^{\ast}=\mathbf{0}, \forall k \\
    \zeta ^{\ast}\left( \sum_{k=1}^K{\left\| \mathbf{w}_{k}^{\ast} \right\| ^2}-P_T \right) =0.
  \end{array} \right. 
\end{align}

Similarly, the Lagrangian function of $\mathcal{P}_4$ is 
\begin{align}
  \mathcal{L}_{2}\left(\left\{\mathbf{w}_{k}^q\right\}_{k=1}^{K},{\boldsymbol{\theta  }},\boldsymbol{\varsigma } \right)&=R_{s}^{ub2}\left(\left\{\mathbf{w}_{k}^q\right\}_{k=1}^{K},{\boldsymbol{\theta}}\right) \nonumber \\
&\quad +\sum_{n=1}^N\varsigma_n\left({\theta}_n^\dagger{\theta}_n-1\right),
\end{align}
where $\boldsymbol{\varsigma } = \left[\varsigma_1 \geq 0,\cdots, \varsigma_N \geq 0\right] $ is the dual variable  in \eqref{P4b}. Then the KKT condition for $\boldsymbol{\theta}$ can be expressed as follows:
\begin{align}
  \left\{ 
    \begin{array}{l}
    \nabla_{{\theta}_{n}}R_{s}^{ub2}\left(\left\{\mathbf{w}_{k}^{q}\right\}_{k=1}^{K},{\boldsymbol{\theta}}^{\ast}\right)+\varsigma_{n}^{\ast}{\theta}_{n}^{\ast}=0,\forall n, \\
    \varsigma_{n}^{\ast}\left(\left({\theta}_{n}^{\ast}\right)^{\dagger}{\theta}_{n}^{\ast}-1\right)=0,\forall n.
  \end{array} \right. 
\end{align}
Since $\left\{\mathbf{w}_{k}^{\ast}\right\}_{k=1}^{K}$  and $\boldsymbol{\theta}^{\ast}$ are the optimal solutions of $\mathcal{P}_2$ and $\mathcal{P}_5$, respectively, thus the above KKT conditions are satisfied. Therefore, the converged solution $\left(\left\{\mathbf{w}_{k}^{\ast}\right\}_{k=1}^{K},{\boldsymbol{\theta}}^{\ast}\right)$ is a KKT point.

\end{document}